\documentclass[twoside]{article}
\usepackage{amsmath, amssymb, amsthm}
\usepackage{mathtools}
\usepackage{graphicx}
\usepackage{hyperref}
\usepackage{geometry}
\usepackage{array}
\usepackage[T1]{fontenc}
\usepackage[utf8]{inputenc}
\usepackage{tikz}
\usetikzlibrary{arrows.meta,calc,positioning,fit,backgrounds}

\tikzset{
    setUniverse/.style={draw=black, thick, rounded corners, inner sep=12pt},
    setClosure/.style={draw=black, double, very thick, rounded corners, inner sep=8pt},
    setLog/.style={draw=black!75, ultra thick, rounded corners, inner sep=10pt},
    setAtoms/.style={draw=black, thick, rounded corners, inner sep=16pt},
    setRedun/.style={draw=black, ultra thick, dashed, rounded corners, inner sep=16pt},
    setRecon/.style={draw=black, ultra thick, dotted, rounded corners, inner sep=12pt},
    elemA/.style={circle, draw=black, semithick, minimum size=3.8mm, inner sep=0pt, fill=black!15},
    elemA2/.style={circle, draw=black, semithick, minimum size=3.8mm, inner sep=0pt, fill=black!30},
    elemA3/.style={circle, draw=black, semithick, minimum size=3.8mm, inner sep=0pt, fill=black!45},
    elemJ/.style={circle, draw=black, semithick, minimum size=3.8mm, inner sep=0pt, fill=white,
                  path picture={\draw[black,semithick] (path picture bounding box.south west)--(path picture bounding box.north east)
                                (path picture bounding box.north west)--(path picture bounding box.south east);}},
    elemC/.style={rectangle, draw=black, semithick, minimum size=3.4mm, inner sep=0pt, fill=gray!35, rounded corners=0.6pt},
    elemH/.style={circle, draw=black, thick, dashed, minimum size=4.2mm, inner sep=0pt, fill=black!15},
    relDer/.style={-{Stealth[length=2.2mm,width=1.8mm]}, thick, draw=black, shorten >=1.5pt, shorten <=1.5pt},
    relGen/.style={-{Latex[length=2.4mm,width=1.9mm]}, very thick, draw=black!70, dash pattern=on 4.5pt off 2.5pt, shorten >=1.5pt, shorten <=1.5pt}
}

\tikzset{
  box/.style={
    draw, rounded corners, align=center,
    inner sep=5pt, minimum height=10mm,
    text=black
  },
  rdbox/.style={
    draw, dashed, rounded corners, align=center,
    inner sep=4pt, minimum height=8mm,
    fill=gray!30, text=black,
    font=\scriptsize
  },
  arr/.style={->, >=Latex, thick, shorten >=2pt, shorten <=2pt},
  outarr/.style={arr, rounded corners},
  rdarr/.style={->, >=Latex, dashed, thin, draw=gray!70,
                shorten >=1pt, shorten <=1pt},
  rdflow/.style={->, >=Latex, dashed, thin, draw=black,
                 shorten >=1pt, shorten <=1pt, rounded corners},
  rdbi/.style={<->, >=Latex, dashed, thin, draw=black,
               shorten >=1pt, shorten <=1pt, rounded corners}
}

\newcommand{\inhArc}[3][]{
    \draw[inh,#1] #2 -- #3; 
    \node[inhmark] at ($#3!0.03!#2$) {};
}

\usepackage{pgfplots}
\pgfplotsset{compat=1.18}
\usepgfplotslibrary{groupplots}
\usepackage{booktabs}
\usepackage{enumitem}

\DeclareMathOperator{\Fr}{Fr}

\newtheorem{theorem}{Theorem}[section]
\newtheorem{definition}[theorem]{Definition}
\newtheorem{corollary}[theorem]{Corollary}
\newtheorem{lemma}[theorem]{Lemma}
\newtheorem{proposition}[theorem]{Proposition}
\newtheorem{assumption}[theorem]{Assumption}

\newcommand{\Cn}{\mathrm{Cn}}
\newcommand{\Atom}{\mathrm{Atom}}
\newcommand{\sem}{\mathrm{sem}}

\newcommand{\E}{\mathbb{E}}
\newcommand{\1}{\mathbf{1}}

\title{Closure-Preserving Rate--Distortion for Reversible Logging}
\author{Jianfeng~Xu \\ Koguan School of Law, China Institute for Smart Justice,\\
School of Computer Science, Shanghai Jiao Tong University,\\ Shanghai 200030, China. \\ \texttt{xujf@sjtu.edu.cn}}
\date{}

\begin{document}
\maketitle

\begin{abstract}
We study semantic compression of reversible-execution evidence for rollback reasoning. 
A run is a finite fact base; rollback semantics are modeled by a monotone closure operator induced by function-free Horn rules. 
A single edit replaces one fact by another; fidelity is the Jaccard discrepancy of the resulting closures, yielding a finite-alphabet distortion for rate–distortion analysis. 
A deterministic deletion scan decomposes the log into an irredundant core—preserving the closure—and a redundant remainder. 
Under admissible reconstructions (facts entailed by the original log), redundant facts are distortion-invisible, reducing the semantic rate–distortion function to a core-only optimization scaled by the core probability mass. 
At zero distortion, the optimal rate is a hypergraph entropy induced by overlaps of zero-distortion reconstruction sets on the core. 
We introduce a rollback-task loss based on a rollback observable, deriving parallel endpoint and factorization laws. 
The framework is instantiated on reversible causal nets and discussed in the event-structure view, showing how reversing disciplines yield different cores and compression frontiers. 
Numerical evaluation uses Blahut–Arimoto to design single-letter test channels and Monte Carlo reconstruction to assess end-to-end degradation at the log level.
\end{abstract}


\section{Introduction}
\label{sec:intro}

Reversible computation studies models in which executed actions may be undone,
motivated by fault recovery, debugging, and reversible views of concurrent or
biochemical dynamics \cite{toffoli1977computation,fredkin1982conservative,
danos2004reversible,ulidowski2018reversing}.
In concurrent settings, rollback is inherently \emph{discipline-sensitive}:
which backward steps are legal depends on causal dependencies, conflict,
prevention evidence, and sometimes reverse-causality constraints
\cite{lanese2024axiomatic,aubert2025independence,melgratti2024reversible}.
Different reversing disciplines demand different auxiliary evidence to be retained.

This paper asks a quantitative version of that requirement:
\emph{how many bits are necessary to store execution evidence so that a
reconstruction preserves rollback-relevant consequences?}
Classical rate--distortion theory provides a natural framework: given a source
and a fidelity criterion, it characterizes the minimum rate needed to achieve a
prescribed average distortion \cite{shannon1959coding,cover2006elements,
csizsar2011information}.
Rollback fidelity, however, is not naturally symbolwise (e.g., Hamming distance on
logged atoms): a local edit can enable or disable many rollback judgments through
the underlying semantics, while some facts may be redundant because they are
derivable from others.
A purely symbolwise distortion therefore over-penalizes semantically redundant
facts and fails to reflect global rollback effects.

We align compression with rollback semantics through a closure-based interface.
A single run is represented as a finite fact base \(S_O\subseteq\mathbb{S}\),
and rollback-relevant meaning is modeled by a monotone closure operator
\(\Cn_{\mathrm{rev}}\) induced by an effective rule system (e.g., function-free
Horn rules / Datalog) \cite{abiteboul1995foundations,krotzsch2025modern}.
A designated rollback observable is then a projection of closure consequences:
\[
\mathsf{RB}(S)=\Cn_{\mathrm{rev}}(S)\cap\mathbb{Q}_{\mathrm{rb}},
\]
for a fixed query vocabulary \(\mathbb{Q}_{\mathrm{rb}}\).
We study the effect of single-symbol edits of the fixed log \(S_O\): replacing
one fact \(s\in S_O\) by a reconstruction symbol \(\hat{s}\in\hat{\mathbb{S}}\)
forms an edited log
\(S_O[s\leftarrow \hat{s}]=(S_O\setminus\{s\})\cup\{\hat{s}\}\).
Distortion is measured by the induced Jaccard-type discrepancy between closures,
yielding a bounded finite-alphabet distortion suitable for Shannon rate--distortion
analysis.

Our information-theoretic object is the rate--distortion function of i.i.d.\ draws
from a fixed reference log: fix \(S_O\), sample \(S\sim P_O\) supported on \(S_O\),
and design a test channel \(P_{\hat{S}\mid S}\) over a reconstruction alphabet
\(\hat{\mathbb{S}}\).
This isolates the semantic redundancy structure of a given run and yields clean
endpoint and factorization laws; it is not a probabilistic model over complete
executions.

Closure-preserving fidelity makes redundancy explicit.
Many logged facts are consequences of others under \(\Cn_{\mathrm{rev}}\) and can
be deleted or replaced without changing rollback semantics.
A deterministic deletion scan (Definition~\ref{def:core}) decomposes \(S_O\) into
an irredundant \emph{core} \(A\) and a redundant remainder \(J\).
Under an admissible reconstruction alphabet
\(\hat{\mathbb{S}}\subseteq\Cn_{\mathrm{rev}}(S_O)\), every redundant fact is
distortion-invisible, so both the zero-distortion limit and the full
rate--distortion curve reduce to a core-only optimization scaled by the core mass
\(P_A=P_O(A)\).
This ``core reduction'' principle, closely related to closure-based
rate--distortion for deductive sources \cite{xu2026rate}, yields a quantitative
compression frontier that depends on the rollback discipline through
\(\Cn_{\mathrm{rev}}\).

We instantiate \(\Cn_{\mathrm{rev}}\) on reversible causal nets (RCN) and
reversible prime event structures (rPES) by defining discipline-indexed monotone
closures for causal, cause-respecting, and inverse-causal semantics
\cite{melgratti2024reversible,melgratti2025relating}.
This makes discipline sensitivity quantitative: different disciplines induce
different cores, hence different rate--distortion trade-offs.
We complement the theory with a numerical evaluation pipeline that matches the
single-log viewpoint: for fixed logs and discipline-indexed closures, we construct
finite alphabets and distortion/loss matrices, compute Blahut--Arimoto optimal
single-letter test channels \cite{blahut1972computation,arimoto1972algorithm},
and lift the resulting channels to set-valued logs by independent per-fact
reconstruction followed by set union.
Because set semantics introduce nonlinearities (duplicates collapsing under
union), we report end-to-end log-level rollback degradation via Monte Carlo
reconstruction.

The paper makes four contributions.

\begin{enumerate}
    \item \textbf{Semantic rate--distortion for reversible logging.}
    We formulate reversible logging as a Shannon rate--distortion problem under a
    closure-preserving fidelity criterion, adapting the general deductive-source
    framework of \cite{xu2026rate} to the reversible setting by instantiating
    \(\Cn_{\mathrm{rev}}\) and defining a single-symbol edit distortion
    \(d_{\Cn}\) aligned with rollback semantics.

    \item \textbf{Rollback-task loss and task-aligned RD.}
    We introduce a rollback-task loss \(\ell_{\mathrm{rb}}\) that counts how many
    rollback-relevant judgments change under an edit, prove a core-only
    factorization for \(R_{\mathrm{rb}}(L)\), and establish a perfect-rollback
    hypergraph-entropy law \(R_{\mathrm{rb}}(0)=P_A H_{\Gamma_0^{\mathrm{rb}}}(\pi_A)\).
    We further show that perfect rollback in this task sense is never harder than
    perfect closure preservation, \(R_{\mathrm{rb}}(0)\le R_{\mathrm{sem}}(0)\),
    with strict separation possible under reduced query sets.

    \item \textbf{Discipline-indexed RCN/rPES instantiation.}
    Building on the RCN/rPES correspondence \cite{melgratti2024reversible}, we
    define monotone Horn closures indexed by reversing disciplines and prove that
    causal/cause-respecting closures yield frontier-shaped cores, while
    inverse-causal closures can expand the core to include causal ancestors.
    Different disciplines thus induce different compression frontiers.

    \item \textbf{End-to-end numerical evaluation at the log level.}
    We implement a pipeline that combines Blahut--Arimoto single-letter
    optimization with Monte Carlo reconstruction of set-valued logs, quantifying
    discipline sensitivity, utility at matched rate points, and rate savings
    under fixed quality targets.
\end{enumerate}

Figure~\ref{fig:revcomp-rd-rect-loop} gives a high-level overview of the overall semantic interface and the rate--distortion design loop employed in this paper. It shows how the reversible logging pipeline, the closure-based fidelity target \(\Cn_{\mathrm{rev}}\), the two distortion measures (\(d_{\Cn}\) for closure preservation and \(\ell_{\mathrm{rb}}\) for rollback tasks), and the canonical core reduction \(S_O=A\uplus J\) are integrated into a finite-alphabet RD problem that can be evaluated numerically (e.g., via Blahut--Arimoto).

\begin{figure}[t]
\centering
\makebox[\linewidth][c]{
\begin{tikzpicture}[font=\small]
\def\hgap{10mm}
\def\vgap{43mm}     
\def\rdgap{5mm}    
\def\rdup{5mm}     
\def\sidegap{10mm}

\node[box, text width=28mm] (run)  {Program\\ execution};
\node[box, text width=28mm, right=\hgap of run] (inst) {Instrumentation\\ \& event capture};
\node[box, text width=28mm, right=\hgap of inst] (log)  {Event log\\ (trace)};

\node[box, text width=28mm] (rb)     at ($(run)+(0,-\vgap)$)   {Rollback\\ recovery};
\node[box, text width=28mm] (replay) at ($(inst)+(0,-\vgap)$)  {Replay\\ \& analysis};
\node[box, text width=28mm] (store)  at (log |- rb.center) {Log management\\ store / index / (optional) compress};

\node[rdbox, text width=28mm, anchor=north] (rd_run)
  at ($(run.south)+(0,-\rdgap)$) {Fidelity target\\(\(\Cn_{\mathrm{rev}}\)-invariance)};
\node[rdbox, text width=28mm, anchor=north] (rd_inst)
  at ($(inst.south)+(0,-\rdgap)$) {Distortion / loss\\(\(d_{\Cn}\) or \(\ell_{\mathrm{rb}}\))};
\node[rdbox, text width=36mm, anchor=north] (rd_log)
  at ($(log.south)+(0,-\rdgap)$) {Core reduction\\(\(A=\Atom_{\mathrm{rev}}(S_O)\)) \& BA};

\node[rdbox, text width=28mm, anchor=south] (rd_rb)
  at ($(rb.north)+(0,\rdup)$) {Fidelity check\\(consistency)};
\node[rdbox, text width=28mm, anchor=south] (rd_replay)
  at ($(replay.north)+(0,\rdup)$) {Reconstruction\\for queries};
\node[rdbox, text width=28mm, anchor=south] (rd_store)
  at ($(store.center |- rb.north)+(0,\rdup)$) {Rate budget\\(keep / tier)};

\draw[arr] (run) -- (inst);
\draw[arr] (inst) -- (log);
\draw[outarr] (log.east) -- ++(\sidegap,0) |- (store.east);
\draw[arr] (store.west) -- (replay.east);
\draw[arr] (replay.west) -- (rb.east);
\draw[outarr] (rb.west) -- ++(-\sidegap,0) |- (run.west);

\draw[rdarr] (run.south)  -- (rd_run.north);
\draw[rdarr] (inst.south) -- (rd_inst.north);
\draw[rdarr] (log.south)  -- (rd_log.north);

\draw[rdarr] (rd_rb.south)     -- (rb.north);
\draw[rdarr] (rd_replay.south) -- (replay.north);
\draw[rdarr] (rd_store.south)  -- (store.north);

\draw[rdflow] (rd_run.east)  -- (rd_inst.west);
\draw[rdflow] (rd_inst.east) -- (rd_log.west);

\draw[rdbi] (rd_log.south) -- (rd_store.north);

\draw[rdflow] (rd_store.west)   -- (rd_replay.east);
\draw[rdflow] (rd_replay.west)  -- (rd_rb.east);

\draw[rdflow] (rd_rb.north) -- (rd_run.south);

\draw[rdbi] (rd_inst.south) -- (rd_replay.north);

\end{tikzpicture}
}
\caption{Semantic interface for closure-preserving reversible logging.
The outer solid loop is the reversible pipeline (execution \(\rightarrow\) instrumentation \(\rightarrow\) log \(\rightarrow\) storage/indexing \(\rightarrow\) replay \(\rightarrow\) rollback \(\rightarrow\) execution).
The inner dashed loop is the rate--distortion (RD) design layer specialized to this paper: the fidelity target is invariance of the reversible semantic closure \(\Cn_{\mathrm{rev}}\); distortion is measured by the closure-based distortion \(d_{\Cn}\) (or the rollback-task loss \(\ell_{\mathrm{rb}}\)); and coding design exploits the canonical core reduction \(S_O=A\uplus J\) with \(A=\Atom_{\mathrm{rev}}(S_O)\) before numerical evaluation (e.g., Blahut--Arimoto) of the resulting finite-alphabet RD problem.}
\label{fig:revcomp-rd-rect-loop}
\end{figure}
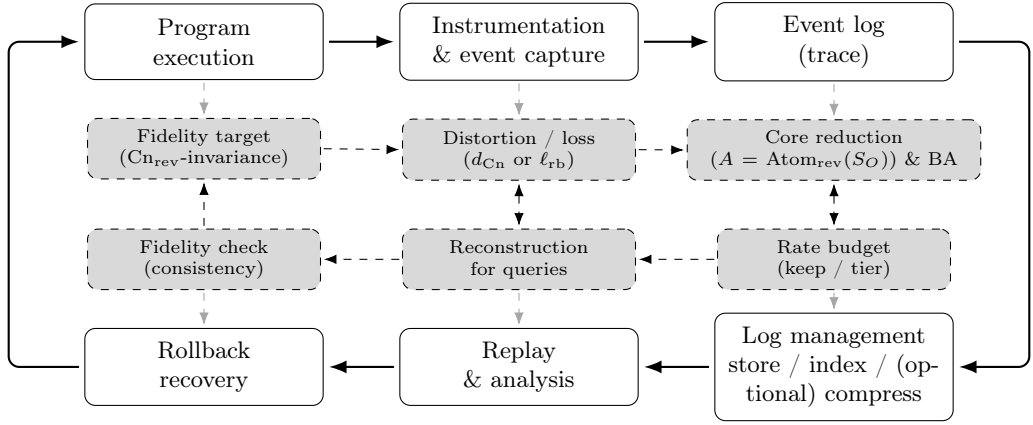

The rest of the article proceeds as follows.
Section~\ref{sec:prelim} defines the fact-base interface, reversible closure, and rollback observables.
Section~\ref{sec:zero} develops semantic rate--distortion under closure-preserving fidelity.
Section~\ref{sec:rb} connects semantic distortion to rollback safety and introduces rollback-task rate--distortion.
Section~\ref{sec:rcn-theory} instantiates discipline-indexed closures on RCN/rPES and analyzes induced cores.
Section~\ref{sec:rcn-exp} reports numerical evaluation combining Blahut--Arimoto design with Monte Carlo reconstruction.
Section~\ref{sec:related} reviews related work, and Section~\ref{sec:conclusion} concludes.

\section{Preliminaries: Fact Bases, Proof Systems, and Rollback Semantics}
\label{sec:prelim}

This section fixes the semantic interface used throughout the paper.
A single execution history is represented as a finite fact base \(S_O\subseteq\mathbb{S}\);
rollback-relevant meaning is captured by a monotone closure operator \(\Cn_{\mathrm{rev}}\)
induced by an effective proof system.
In the rest of the paper, preserving the full closure \(\Cn_{\mathrm{rev}}(S_O)\) serves as a sufficient fidelity criterion for rollback,
while \(\mathsf{RB}(S)\) (Definition~\ref{def:rb-capability}) captures the minimal rollback observable of interest.
Our instantiations follow standard causal-consistent reversibility frameworks (e.g., RCCS/LTSI and reversible Petri-net/event-structure accounts)
and reversible message-passing logs; see, e.g., \cite{lanese2024axiomatic,melgratti2024reversible,oguchi2025revmigo}.

Throughout, rollback-relevant semantics are captured through (projections of)
closure consequences \(\Cn_{\mathrm{rev}}(S)\).
Accordingly, we adopt closure preservation
\(\Cn_{\mathrm{rev}}(\hat{S})=\Cn_{\mathrm{rev}}(S)\) as a principled and
sufficient semantic fidelity criterion for reversible logging, aligning the
reversible setting with closure-based compression.

\subsection{Fact universe and log fact bases}
\label{sec:fact-universe}

Fix a finite set \(\mathcal{E}\) of \emph{forward events}. Let \(\underline{\mathcal{E}}\) be a disjoint set of corresponding
\emph{reverse events}; we write \(\underline{e}\) for the reverse of \(e\in\mathcal{E}\).
In Petri-net settings, events may be transition occurrences; in process-calculus settings, events may be equivalence classes of
transitions (as in LTSI/event-based accounts); in message-passing settings, events may be log primitives.

We assume each event identifier \(e\in \mathcal{E}\cup \underline{\mathcal{E}}\) is represented by a corresponding constant symbol,
so that event occurrences and relations can be recorded as ground facts in \(\mathbb{S}\).

We work with a model-agnostic interface in which the target reversible model induces relations on events, possibly partially defined:
\begin{itemize}
  \item \emph{Causality} \(<\ \subseteq \mathcal{E}\times \mathcal{E}\);
  \item \emph{Conflict} \(\#\ \subseteq \mathcal{E}\times \mathcal{E}\);
  \item \emph{Reverse-causality} \(\prec\ \subseteq \mathcal{E}\times \underline{\mathcal{E}}\) (``\(e\) is a cause of undoing \(u\)'');
  \item \emph{Prevention} \(\triangleright\ \subseteq \mathcal{E}\times \underline{\mathcal{E}}\) (``presence of \(e\) blocks undoing \(u\)'');
  \item Optionally, \emph{sustained causation} \(\ll\ \subseteq \mathcal{E}\times \mathcal{E}\), used in axiomatic frameworks to tie
        \(<\) to prevention constraints.
\end{itemize}
The concrete meaning of these relations is instantiated in Section~\ref{sec:rcn-theory}.

\paragraph{Fact signature and ambient universe.}
Let \(\Sigma_{\mathrm{rev}}\) be a finite relational signature containing at least the predicate symbols:
\[
\mathsf{Fwd}(\cdot),\ \mathsf{Rev}(\cdot),\ \mathsf{Occurs}(\cdot),\ 
\mathsf{Cause}(\cdot,\cdot),\ \mathsf{Conflict}(\cdot,\cdot),\ 
\mathsf{RevCause}(\cdot,\cdot),\ \mathsf{Prevent}(\cdot,\cdot).
\]
(Concrete instantiations may add, e.g., \(\mathsf{Key}(k,\cdot,\cdot)\), FIFO order facts, place/marking facts, etc.)
Fix a finite set of constants containing all event identifiers for the execution under study (and any auxiliary identifiers such as keys).
Let \(\mathbb{S}\) denote the set of all \emph{ground atoms} over \(\Sigma_{\mathrm{rev}}\) and these constants.
Since both the signature and the constant set are finite, the ambient universe \(\mathbb{S}\) is finite.

\begin{definition}[Log fact base]
\label{def:log-fact-base}
A (single-run) \emph{log fact base} is any finite set \(S_O\subseteq\mathbb{S}\).
Intuitively, \(S_O\) contains:
\begin{enumerate}
  \item \emph{Dynamic history facts}, e.g., \(\mathsf{Occurs}(e)\) for executed forward events \(e\),
        and possibly model-dependent facts such as ``a send with key \(k\) occurred''.
  \item \emph{Static semantic facts} needed for rollback reasoning, e.g., \(\mathsf{Cause}(e,e')\),
        \(\mathsf{RevCause}(e,\underline{u})\), \(\mathsf{Prevent}(e,\underline{u})\), and \(\mathsf{Conflict}(e,e')\),
        when these are not globally fixed by the model.
\end{enumerate}
\end{definition}

For readability, we allow two equivalent presentations:
\begin{itemize}
  \item \textbf{All-in-log:} \(S_O\) contains both dynamic and static facts relevant to the run.
  \item \textbf{Fixed background:} a fixed finite set \(\mathcal{B}\subseteq\mathbb{S}\) encodes the static semantics for the program/model,
        and \(S_O\) contains only dynamic history facts; then reasoning is performed over \(\mathcal{B}\cup S_O\).
\end{itemize}
All definitions below work in either view; when needed we will write \(\Cn_{\mathrm{rev}}(S)\) as shorthand for
\(\Cn_{\mathrm{rev}}(\mathcal{B}\cup S)\) in the ``fixed background'' view.

\begin{assumption}[Finite universe]
\label{assm:finiteness}
For every execution under study, the logged fact base \(S_O\subseteq\mathbb{S}\) is finite, and the ambient universe
\(\mathbb{S}\) of ground atoms (over a finite relational signature and a finite constant set) is finite.
\end{assumption}

Because \(\mathbb{S}\) is finite, any function-free Horn closure over \(\mathbb{S}\) is finite and can be computed by iterative
saturation; a precise statement is given in Lemma~\ref{lem:finite-closure-computable} of Section~\ref{sec:closure}.

\subsection{A monotone reversible proof system and semantic closure}
\label{sec:closure}

We assume an effective monotone proof system \(\mathsf{PS}_{\mathrm{rev}}\) over ground facts.
In the concrete instantiations in this paper, \(\mathsf{PS}_{\mathrm{rev}}\) is given as a function-free Horn rule system (Datalog),
so that \(\Cn_{\mathrm{rev}}(S)\) is the least fixed point of the associated immediate-consequence operator.
The concrete rules are model-dependent, but conceptually they (i) propagate derived consequences of history facts and background structure,
and (ii) derive rollback-relevant judgments (e.g., blocking evidence and precedence obligations).
Crucially, we require that the induced closure operator satisfy the Tarski axioms (Assumption~\ref{assm:tarski} below).
For background on Datalog/Horn-rule inference and least-fixed-point semantics, see standard database references (e.g., \cite{abiteboul1995foundations}).

\begin{definition}[Reversible semantic closure]
\label{def:cn-rev}
Fix a monotone proof system \(\mathsf{PS}_{\mathrm{rev}}\) over \(\mathbb{S}\).
For any finite \(S\subseteq \mathbb{S}\), define
\[
\Cn_{\mathrm{rev}}(S) \coloneqq \{ s\in \mathbb{S} : S \vdash_{\mathsf{PS}_{\mathrm{rev}}} s \}.
\]
\end{definition}

When \(\mathsf{PS}_{\mathrm{rev}}\) is given as function-free Horn rules (Datalog), we take \(\Cn_{\mathrm{rev}}(S)\) to be the least fixed point
of the immediate-consequence operator induced by the rules, which coincides with Horn derivability
\cite{abiteboul1995foundations,dantsin2001complexity}.

\begin{lemma}[Finite Horn closure is effectively computable]
\label{lem:finite-closure-computable}
Assume \(\mathbb{S}\) is finite and \(\mathsf{PS}_{\mathrm{rev}}\) is a function-free Horn rule system over \(\mathbb{S}\).
Then for every \(S\subseteq\mathbb{S}\), the closure \(\Cn_{\mathrm{rev}}(S)\) is finite and can be computed by iterative saturation,
i.e., by repeatedly applying rules until a fixed point is reached.
\end{lemma}

\begin{proof}
Because \(\mathbb{S}\) is finite and \(\Cn_{\mathrm{rev}}(S)\subseteq \mathbb{S}\), any increasing sequence of derived fact sets stabilizes
after at most \(|\mathbb{S}|\) strict additions. For Horn rules, the immediate-consequence operator is monotone, hence iterating it from \(S\)
reaches the least fixed point in finitely many steps, which equals the set of derivable facts (see, e.g., \cite{abiteboul1995foundations,dantsin2001complexity}).
\end{proof}

\begin{assumption}[Tarski closure axioms / sufficient conditions]
\label{assm:tarski}
The operator \(\Cn_{\mathrm{rev}}\) satisfies:
(i) \(S\subseteq \Cn_{\mathrm{rev}}(S)\) (reflexivity),
(ii) \(S\subseteq T \Rightarrow \Cn_{\mathrm{rev}}(S)\subseteq \Cn_{\mathrm{rev}}(T)\) (monotonicity),
(iii) \(\Cn_{\mathrm{rev}}(\Cn_{\mathrm{rev}}(S))=\Cn_{\mathrm{rev}}(S)\) (idempotence).

A sufficient (and typical) condition is that \(\mathsf{PS}_{\mathrm{rev}}\) is a function-free Horn rule system
(e.g., Datalog) and \(\Cn_{\mathrm{rev}}(S)\) is its least fixed point.
\end{assumption}

\begin{lemma}[Horn closure implies Tarski axioms]
\label{lem:horn-tarski}
If \(\mathsf{PS}_{\mathrm{rev}}\) is a monotone Horn inference system and \(\Cn_{\mathrm{rev}}(S)\) is defined as the set of all
facts derivable from \(S\) (equivalently, the least fixed point of the immediate-consequence operator),
then \(\Cn_{\mathrm{rev}}\) satisfies the three Tarski closure axioms in Assumption~\ref{assm:tarski}.
\end{lemma}

\begin{proof}
For Horn rules, let \(T_{\mathsf{PS}}\) be the immediate-consequence operator. It is monotone, hence by the Knaster--Tarski fixed-point theorem
its least fixed point \(\mathsf{lfp}(T_{\mathsf{PS}})\) exists \cite{tarski1955lattice}.
By definition, \(\Cn_{\mathrm{rev}}(S)\) is obtained by iterating \(T_{\mathsf{PS}}\) from \(S\) to the least fixed point, so:
(i) \(S\subseteq \Cn_{\mathrm{rev}}(S)\) holds because the iteration starts from \(S\);
(ii) monotonicity holds because \(T_{\mathsf{PS}}\) is monotone in its input;
(iii) idempotence holds because \(\Cn_{\mathrm{rev}}(S)\) is already a fixed point of \(T_{\mathsf{PS}}\), hence closing it again adds no facts.
\end{proof}

\paragraph{Remark (negative enabling conditions and monotone closures).}
Operational presentations of reversibility often phrase enabling conditions
negatively (e.g., ``\(\underline{u}\) is enabled only if no preventer is
present'').
Such conditions are non-monotone if encoded directly as logical inference.
We therefore require \(\Cn_{\mathrm{rev}}\) to be a \emph{monotone} closure
operator and encode \emph{evidence of blocking} as positive derived facts,
e.g.,
\[
\mathsf{Prevent}(e,\underline{u}) \wedge \mathsf{Occurs}(e)
\Rightarrow \mathsf{Blocked}(\underline{u}),
\]
and choose rollback-relevant observables \(\mathbb{Q}_{\mathrm{rb}}\) so that
rollback capability is expressed through such positive evidence (e.g., blockers
and precedence obligations).
This keeps \(\Cn_{\mathrm{rev}}\) within the standard monotone Horn/Datalog
setting while still capturing rollback constraints.

\subsection{Rollback semantics and closure invariance}
\label{sec:rollback-equiv}

The role of \(\Cn_{\mathrm{rev}}\) is to provide a \emph{semantic interface} for rollback.
We select a designated set \(\mathbb{Q}_{\mathrm{rb}}\subseteq \mathbb{S}\) of \emph{rollback-relevant} ground facts
and define rollback capability as the closure projected to \(\mathbb{Q}_{\mathrm{rb}}\).
Typical examples include
\[
\mathsf{Blocked}(\underline{e}),\quad
\mathsf{MustUndoBefore}(\underline{e},\underline{u}),\quad
\mathsf{Maximal}(e),\quad
\mathsf{NonMax}(e),
\]
depending on the chosen reversible model.
Since \(\mathbb{S}\) is finite (Assumption~\ref{assm:finiteness}) and \(\mathbb{Q}_{\mathrm{rb}}\subseteq \mathbb{S}\), the query set \(\mathbb{Q}_{\mathrm{rb}}\) is finite.

\begin{definition}[Rollback capability as a closure observable]
\label{def:rb-capability}
Define the rollback capability induced by a log \(S\subseteq\mathbb{S}\) as
\[
\mathsf{RB}(S)\ \coloneqq\ \Cn_{\mathrm{rev}}(S)\ \cap\ \mathbb{Q}_{\mathrm{rb}}.
\]
\end{definition}

\begin{definition}[Rollback-equivalence (semantic target)]
\label{def:rollback-equiv}
Two logs \(S,T\subseteq \mathbb{S}\) are \emph{rollback-equivalent}, written \(S \equiv_{\mathrm{rb}} T\), if
\[
S \equiv_{\mathrm{rb}} T
\quad\Longleftrightarrow\quad
\mathsf{RB}(S)=\mathsf{RB}(T).
\]
\end{definition}

The next statement explains when rollback-equivalence coincides with full closure equivalence.

\begin{proposition}[When rollback-equivalence matches full closure equivalence]
\label{prop:rb-vs-closure}
If \(\mathbb{Q}_{\mathrm{rb}}=\mathbb{S}\), then for all finite \(S,T\subseteq\mathbb{S}\),
\[
S \equiv_{\mathrm{rb}} T \quad\Longleftrightarrow\quad \Cn_{\mathrm{rev}}(S)=\Cn_{\mathrm{rev}}(T).
\]
More generally, \(\Cn_{\mathrm{rev}}(S)=\Cn_{\mathrm{rev}}(T)\) always implies \(S\equiv_{\mathrm{rb}}T\) for any choice of \(\mathbb{Q}_{\mathrm{rb}}\). 
In particular, full closure equivalence is a sufficient condition for rollback-equivalence for any \(\mathbb{Q}_{\mathrm{rb}}\),
while the converse need not hold when \(\mathbb{Q}_{\mathrm{rb}}\subsetneq \mathbb{S}\).
\end{proposition}

\begin{proof}
If \(\mathbb{Q}_{\mathrm{rb}}=\mathbb{S}\), then \(\mathsf{RB}(S)=\Cn_{\mathrm{rev}}(S)\) by Definition~\ref{def:rb-capability},
so the equivalence is immediate from Definition~\ref{def:rollback-equiv}. The general implication
\(\Cn_{\mathrm{rev}}(S)=\Cn_{\mathrm{rev}}(T)\Rightarrow \mathsf{RB}(S)=\mathsf{RB}(T)\) follows by intersecting both sides with \(\mathbb{Q}_{\mathrm{rb}}\).
\end{proof}

In applications (Section~\ref{sec:rcn-theory}), we choose \(\mathbb{Q}_{\mathrm{rb}}\) so that \(\mathsf{RB}(S)\)
captures the rollback-relevant consequences that are representable by positive evidence in a monotone setting
(e.g., blocking evidence and precedence obligations).
Thus, preserving \(\Cn_{\mathrm{rev}}(S_O)\) is a sufficient condition for preserving rollback capability,
while preserving \(\mathsf{RB}(S_O)\) is the minimal requirement when only rollback observables matter.
With this interface in place, we can apply closure-preserving rate--distortion ideas in the spirit of \cite{xu2026rate}.

\begin{assumption}[Adequacy of the rollback observable]
\label{assm:rb-adequacy}
For the chosen discipline and rollback query vocabulary \(\mathbb{Q}_{\mathrm{rb}}\), equality
\[
\mathsf{RB}(S)=\mathsf{RB}(T)
\]
implies coincidence of the rollback judgments of interest in the target application
(e.g., the enabled reverse steps and the precedence obligations queried by the debugger or recovery mechanism).
\end{assumption}

Without Assumption~\ref{assm:rb-adequacy}, the results of Section~\ref{sec:rb} should be read as statements about
preservation of the selected rollback observable \(\mathsf{RB}(\cdot)\), rather than as a full operational equivalence theorem
for the underlying reversible semantics.
Under Assumption~\ref{assm:rb-adequacy}, these observable-preservation results specialize to rollback-safety guarantees
for the judgments represented by \(\mathbb{Q}_{\mathrm{rb}}\).

\section{Semantic Rate--Distortion under Closure-Preserving Fidelity}
\label{sec:zero}

This section instantiates a Shannon-style rate--distortion problem where \emph{distortion} is measured by preservation of reversible semantic closure.
The construction follows the spirit of closure-based compression for finite deductive systems (cf.\ \cite{xu2026rate}),
with the reversible closure \(\Cn_{\mathrm{rev}}\) playing the role of deductive closure.

\subsection{Closure fidelity and single-symbol distortion}
\label{sec:distortion}

We adopt the ``single stored log'' viewpoint: a finite log \(S_O\subseteq\mathbb{S}\) is given, and a random symbol is sampled from it.

\begin{assumption}[Symbol source from a fixed log]
\label{assm:source}
Fix a reference log \(S_O\subseteq\mathbb{S}\).
Let \(S\sim P_O\) be a random \emph{fact symbol} taking values in \(S_O\).
\end{assumption}

\paragraph{Modeling scope.}
The source variable \(S\sim P_O\) is a \emph{single-letter surrogate} extracted from one fixed reference log \(S_O\).
It is not a probabilistic model over complete executions.
Accordingly, the operational reading of \(R_{\sem}(D)\) is: first fix one execution log \(S_O\); then study the asymptotically
optimal coding rate for i.i.d.\ draws of log symbols from \(P_O\), while evaluating distortion relative to the same reference log \(S_O\).
This isolates the semantic redundancy structure of a fixed log and yields a finite-alphabet rate--distortion problem.

\begin{assumption}[Admissible reconstruction alphabet]
\label{assm:admissible}
The reconstruction alphabet \(\hat{\mathbb{S}}\subseteq\mathbb{S}\) is nonempty and \emph{admissible} in the sense that
\[
\hat{\mathbb{S}} \subseteq \Cn_{\mathrm{rev}}(S_O).
\]
\end{assumption}

Admissibility means reconstructions are drawn from facts already semantically entailed by the original log.
This matches the reversible-logging intuition: we only reconstruct facts that are semantically compatible with the original run.

\begin{definition}[Closure Jaccard fidelity]
\label{def:fidelity}
For finite \(S,T\subseteq\mathbb{S}\), define
\[
\mathsf{F}_{\Cn}(S,T) \coloneqq
\frac{|\Cn_{\mathrm{rev}}(S)\cap \Cn_{\mathrm{rev}}(T)|}{|\Cn_{\mathrm{rev}}(S)\cup \Cn_{\mathrm{rev}}(T)|},
\quad \text{with } 0/0 \coloneqq 1.
\]
The convention \(0/0\coloneqq 1\) corresponds to treating two empty closures as perfectly matching.
\end{definition}

\begin{definition}[Single-symbol closure distortion]
\label{def:distortion}
For \(s\in S_O\) and \(\hat{s}\in \hat{\mathbb{S}}\), define
\[
d_{\Cn}(s,\hat{s})
\coloneqq 1-\mathsf{F}_{\Cn}\bigl(S_O,\ (S_O\setminus\{s\})\cup\{\hat{s}\}\bigr).
\]
\end{definition}

The quantity \(d_{\Cn}(s,\hat s)\) is defined \emph{relative to the fixed reference log} \(S_O\): changing \(S_O\) changes the edited log
\((S_O\setminus\{s\})\cup\{\hat s\}\) and hence the induced closures being compared.
When helpful, we may write \(d_{\Cn}^{S_O}(s,\hat s)\) to make this dependence explicit.
This log dependence is essential for the core decomposition \(S_O=A\uplus J\) and for the strict zero-distortion invisibility of the redundant
part \(J\) under admissibility (Proposition~\ref{prop:redundant-zero}).

\begin{lemma}[Basic properties of \(d_{\Cn}\)]
\label{lem:dist-basic}
For all \(s\in S_O\) and \(\hat{s}\in\hat{\mathbb{S}}\), \(0\le d_{\Cn}(s,\hat{s})\le 1\).
Moreover, \(d_{\Cn}(s,\hat{s})=0\) if and only if
\(\Cn_{\mathrm{rev}}(S_O)=\Cn_{\mathrm{rev}}((S_O\setminus\{s\})\cup\{\hat{s}\})\).
\end{lemma}
\begin{proof}
Both claims follow directly from Definition~\ref{def:fidelity} (Jaccard similarity lies in \([0,1]\) under the \(0/0\coloneqq 1\) convention)
and Definition~\ref{def:distortion}.
\end{proof}

By Lemma~\ref{lem:dist-basic}, \(d_{\Cn}(s,\hat{s})=0\) exactly when the single-symbol replacement preserves the full closure.

The distortion \(d_{\Cn}(s,\hat{s})\) measures the semantic impact of modifying a log by \emph{one} symbol while keeping the remaining facts fixed.
This is the natural ``local'' perturbation in closure-based compression: although the edit is local at the syntactic level, its effect is global
at the semantic level through \(\Cn_{\mathrm{rev}}\). This matches the reversible-logging viewpoint where a small change in the recorded history
may enable/disable many rollback-relevant consequences.

\subsection{Irredundant core and redundant part}
\label{sec:core}

\begin{assumption}[Public canonical order]
\label{assm:order}
Encoder and decoder share a public canonical total order \(\preceq\) over \(\mathbb{S}\).
\end{assumption}

\begin{definition}[Irredundant core via deletion scan]
\label{def:core}
Define \(A=\Atom_{\mathrm{rev}}(S_O)\) by the deterministic deletion scan:
start from \(A\leftarrow S_O\); scan \(s\in S_O\) in \(\preceq\)-order and delete \(s\) if
\(s\in \Cn_{\mathrm{rev}}(A\setminus\{s\})\).
Let \(J\coloneqq S_O\setminus A\).
\end{definition}

\begin{theorem}[Core properties]
\label{thm:core-properties}
Under Assumptions~\ref{assm:tarski} and~\ref{assm:order}, the core \(A\) satisfies:
\begin{enumerate}
  \item \(\Cn_{\mathrm{rev}}(A)=\Cn_{\mathrm{rev}}(S_O)\).
  \item \(\forall a\in A,\ a\notin \Cn_{\mathrm{rev}}(A\setminus\{a\})\) (irredundancy).
  \item \(A\) is uniquely determined by \((S_O,\Cn_{\mathrm{rev}},\preceq)\).
\end{enumerate}
\end{theorem}

\begin{proof}
Let \(s_1\prec s_2\prec \cdots \prec s_m\) be the \(\preceq\)-increasing enumeration of \(S_O\).
Set \(A^{(0)}\coloneqq S_O\). For \(k=1,\dots,m\), define \(A^{(k)}\) by
\[
A^{(k)}\coloneqq
\begin{cases}
A^{(k-1)}\setminus\{s_k\}, & \text{if } s_k\in \Cn_{\mathrm{rev}}\!\bigl(A^{(k-1)}\setminus\{s_k\}\bigr),\\
A^{(k-1)}, & \text{otherwise.}
\end{cases}
\]
Let \(A\coloneqq A^{(m)}\).

\noindent\textbf{(1)} If \(s_k\) is deleted, then by assumption
\(s_k\in \Cn_{\mathrm{rev}}(A^{(k)})\). Since \(A^{(k-1)}=A^{(k)}\cup\{s_k\}\subseteq \Cn_{\mathrm{rev}}(A^{(k)})\),
monotonicity and idempotence yield
\[
\Cn_{\mathrm{rev}}(A^{(k-1)})\subseteq \Cn_{\mathrm{rev}}(\Cn_{\mathrm{rev}}(A^{(k)}))
=\Cn_{\mathrm{rev}}(A^{(k)}).
\]
The reverse inclusion \(\Cn_{\mathrm{rev}}(A^{(k)})\subseteq \Cn_{\mathrm{rev}}(A^{(k-1)})\) follows from
\(A^{(k)}\subseteq A^{(k-1)}\) and monotonicity. Hence
\(\Cn_{\mathrm{rev}}(A^{(k)})=\Cn_{\mathrm{rev}}(A^{(k-1)})\).
If \(s_k\) is kept, then \(A^{(k)}=A^{(k-1)}\) and the closure is trivially unchanged.
By induction over \(k\), \(\Cn_{\mathrm{rev}}(A)=\Cn_{\mathrm{rev}}(S_O)\).

\noindent\textbf{(2)} Fix \(a\in A\). Suppose towards a contradiction that
\(a\in \Cn_{\mathrm{rev}}(A\setminus\{a\})\).
Let \(A^{(t)}\) be the working set at the moment when \(a\) is scanned.
Since the scan only deletes elements, we have \(A\subseteq A^{(t)}\), hence
\(A\setminus\{a\}\subseteq A^{(t)}\setminus\{a\}\).
By monotonicity,
\[
a\in \Cn_{\mathrm{rev}}(A\setminus\{a\})
\ \Rightarrow\
a\in \Cn_{\mathrm{rev}}(A^{(t)}\setminus\{a\}),
\]
so \(a\) would be deleted when scanned, contradicting \(a\in A\).
Therefore \(a\notin \Cn_{\mathrm{rev}}(A\setminus\{a\})\).

\noindent\textbf{(3)} The scan is deterministic given \((S_O,\Cn_{\mathrm{rev}},\preceq)\), hence \(A\) is unique.
\end{proof}

\begin{proposition}[Redundant facts have zero distortion against any admissible reconstruction]
\label{prop:redundant-zero}
Assume Assumption~\ref{assm:admissible}. For any \(j\in J\) and any \(\hat{s}\in \hat{\mathbb{S}}\),
\[
d_{\Cn}(j,\hat{s})=0.
\]
\end{proposition}

\begin{proof}
Fix \(j\in J\). By Theorem~\ref{thm:core-properties}(1), \(\Cn_{\mathrm{rev}}(S_O)=\Cn_{\mathrm{rev}}(A)\).
Since \(j\notin A\), we have \(A\subseteq S_O\setminus\{j\}\). Thus, by monotonicity,
\[
\Cn_{\mathrm{rev}}(S_O)=\Cn_{\mathrm{rev}}(A)\subseteq \Cn_{\mathrm{rev}}(S_O\setminus\{j\})\subseteq \Cn_{\mathrm{rev}}(S_O),
\]
so \(\Cn_{\mathrm{rev}}(S_O\setminus\{j\})=\Cn_{\mathrm{rev}}(S_O)\).

Now let \(\hat{s}\in\hat{\mathbb{S}}\). By admissibility, \(\hat{s}\in \Cn_{\mathrm{rev}}(S_O)=\Cn_{\mathrm{rev}}(S_O\setminus\{j\})\).
Hence \((S_O\setminus\{j\})\cup\{\hat{s}\}\subseteq \Cn_{\mathrm{rev}}(S_O\setminus\{j\})\), and by idempotence,
\[
\Cn_{\mathrm{rev}}((S_O\setminus\{j\})\cup\{\hat{s}\})=\Cn_{\mathrm{rev}}(S_O\setminus\{j\})=\Cn_{\mathrm{rev}}(S_O).
\]
Therefore the closures compared in Definition~\ref{def:distortion} coincide, so \(d_{\Cn}(j,\hat{s})=0\).
\end{proof}

\subsection{Zero-distortion hypergraph and entropy law}
\label{sec:hypergraph}

\begin{definition}[Zero-distortion reconstruction sets]
\label{def:zero-set}
For \(s\in S_O\), define
\[
R_0(s)\coloneqq \{ \hat{s}\in\hat{\mathbb{S}} : d_{\Cn}(s,\hat{s})=0\}.
\]
\end{definition}

\begin{definition}[Core distribution]
\label{def:core-dist}
Let \(P_A\coloneqq P_O(A)\). When \(P_A>0\), define \(\pi_A(\cdot)\coloneqq P_O(\cdot \mid A)\).
Let \(A^\star\sim \pi_A\) denote a random core symbol.
\end{definition}

\begin{definition}[Zero-distortion confusability hypergraph and hypergraph entropy]
\label{def:gamma0}
Define the hypergraph \(\Gamma_0\) on vertex set \(A\) by
\[
\Gamma_0 \coloneqq \Bigl\{ W\subseteq A: W\neq\emptyset,\ \bigcap_{a\in W}R_0(a)\neq\emptyset \Bigr\}.
\]
Define the hypergraph entropy
\[
H_{\Gamma_0}(\pi_A)\coloneqq 
\inf_{\substack{P_{W\mid A^\star}:\\ W\in\Gamma_0\ \text{a.s.},\ A^\star\in W\ \text{a.s.}}}
I(A^\star;W),
\]
where \(W\) is a random subset of \(A\) taking values in \(\Gamma_0\).
\footnote{When \(\Gamma_0\) arises from an undirected graph \(G\) on \(A\) (e.g., \(\Gamma_0\) is the downward closure of the clique
hyperedges induced by a pairwise confusability relation), the quantity \(H_{\Gamma_0}(\pi_A)\) reduces to K\"orner's graph-entropy
formulation (up to the standard convention of whether \(G\) encodes confusability or its complement).
We use the hypergraph form to accommodate genuinely higher-order overlaps among zero-distortion reconstruction sets.}
\end{definition}

This quantity is a hypergraph-valued analogue of confusability-based graph entropy (K\"orner) for zero-error settings \cite{korner1971coding,csiszar2011information}.

If the feasible set in the infimum defining \(H_{\Gamma_0}(\pi_A)\) is empty, we set
\[
H_{\Gamma_0}(\pi_A)\coloneqq +\infty.
\]
Equivalently, \(H_{\Gamma_0}(\pi_A)<\infty\) only if every core symbol \(a\in A\) has at least one zero-distortion
reconstruction, i.e.,
\[
R_0(a)\neq\emptyset \qquad \text{for all } a\in A.
\]
The same convention will be used for the rollback hypergraph entropy
\(H_{\Gamma^{\mathrm{rb}}_0}(\pi_A)\) in Section~\ref{sec:rb-zero}.

For each reconstruction symbol \(\hat s\in \hat{\mathbb{S}}\), define the induced witness set
\[
W_{\hat s}\coloneqq \{a\in A:\ \hat s\in R_0(a)\}.
\]
Then \(W_{\hat s}\in \Gamma_0\) whenever it is nonempty, since \(\hat s\in\bigcap_{a\in W_{\hat s}}R_0(a)\).
Conversely, \(W\in\Gamma_0\) iff there exists some \(\hat s\in\hat{\mathbb{S}}\) such that \(W\subseteq W_{\hat s}\).
Thus \(\Gamma_0\) is downward closed, and it can be visualized by representative witness hyperedges \(W_{\hat s}\)
(cf.\ Figure~\ref{fig:sem-rd-zero}, right).

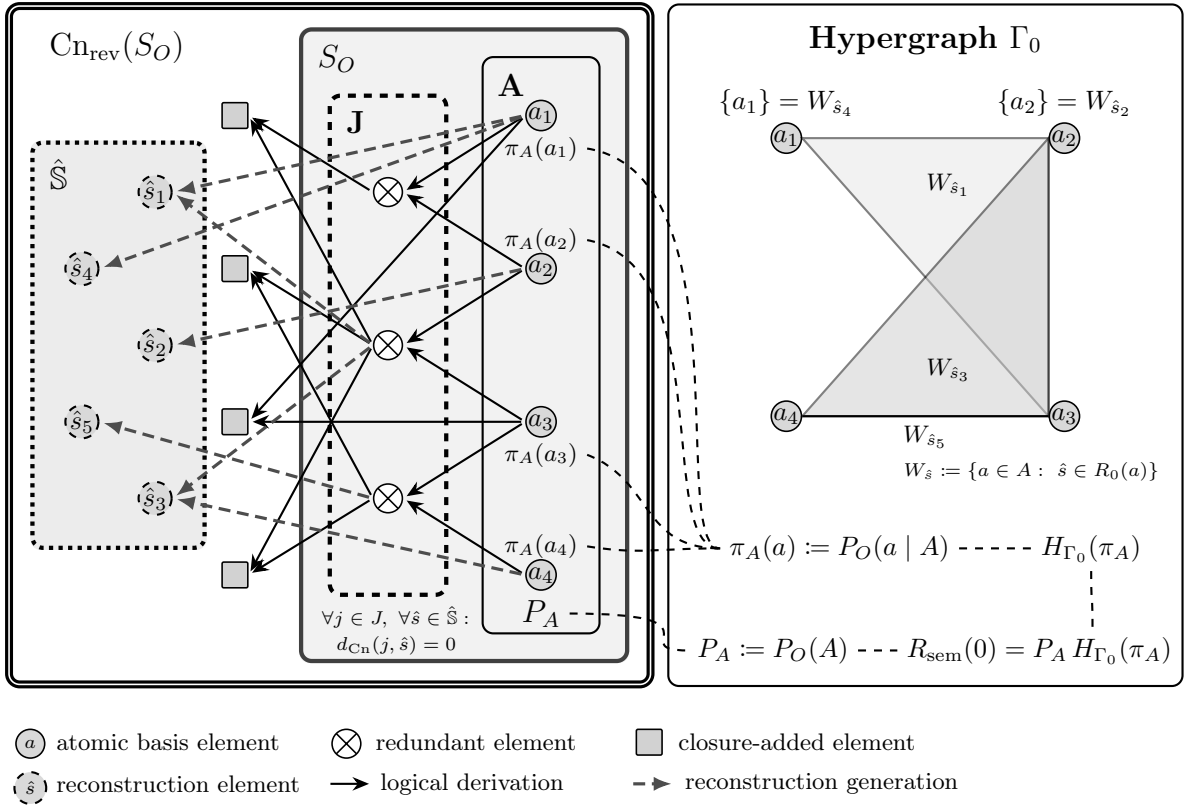
\begin{figure}[t]
\def\LeftMarginShift{5mm}

\begin{tikzpicture}[x=0.88cm, y=1.125cm]

    \begin{scope}[local bounding box=LeftScope, xshift=\LeftMarginShift, xscale=-1]

        \node[elemA]  (a1) at (0.0, 2.7) {\(a_1\)};
        \node[elemA] (a2) at (0.0, 0.9) {\(a_2\)};
        \node[elemA] (a3) at (0.0,-0.9) {\(a_3\)};
        \node[elemA]  (a4) at (0.0,-2.7) {\(a_4\)};

        \node[anchor=north, font=\small\bfseries] (pia1) at ($(a1.south)+(0,0.05)$) {\(\pi_A(a_1)\)};
        \node[anchor=north, font=\small\bfseries] (pia2) at ($(a2.north)+(0,0.4)$) {\(\pi_A(a_2)\)};
        \node[anchor=north, font=\small\bfseries] (pia3) at ($(a3.south)+(0,0.05)$) {\(\pi_A(a_3)\)};        
        \node[anchor=north, font=\small\bfseries] (pia4) at ($(a4.north)+(0,0.4)$) {\(\pi_A(a_4)\)};

        \node[anchor=north, font=\large \bfseries] (pa) at ($(a4.south)+(0,0)$) {\(P_A\)};

        \node[elemJ] (j1) at (2.3, 1.8) {};
        \node[elemJ] (j2) at (2.3, 0.0) {};
        \node[elemJ] (j3) at (2.3,-1.8) {};

        \node[elemC] (c1) at (4.6, 2.7) {};
        \node[elemC] (c2) at (4.6, 0.9) {};
        \node[elemC] (c3) at (4.6,-0.9) {};
        \node[elemC] (c4) at (4.6,-2.7) {};

        \node[elemH] (lh01) at (5.8, 1.8) {\(\hat{s}_1\)};
        \node[elemH] (lh02) at (5.8, 0.0) {\(\hat{s}_2\)};
        \node[elemH] (lh03) at (5.8,-1.8) {\(\hat{s}_3\)};

        \node[elemH] (lhp1) at (6.9, 0.9) {\(\hat{s}_4\)};
        \node[elemH] (lhp2) at (6.9,-0.9) {\(\hat{s}_5\)};

        \node[setAtoms, fit=(a1)(a4)] (Aset) {};

        \node[inner sep=0pt, fit=(j1)(j3)] (J_temp) {};
        \coordinate (top_j_base) at ($(J_temp.north west)!(a1)!(J_temp.south west)$);
        \coordinate (bot_j_base) at ($(J_temp.north west)!(a4)!(J_temp.south west)$);
        \coordinate (top_j_final) at ($(top_j_base)+(0,-3mm)$);
        \coordinate (bot_j_final) at ($(bot_j_base)+(0,3mm)$);
        \node[setRedun, fit=(top_j_final)(bot_j_final)(J_temp)] (Jset) {};

        \node[setLog, fit=(Aset)(Jset)] (Sset) {};
        \node[setRecon, fit=(lh01)(lh03)(lhp1)(lhp2)] (Hset) {};
        \node[setClosure, fit=(Sset)(c1)(c4)(Hset)] (Cset) {};

        \node[anchor=north west, font=\large\bfseries] at ([shift={(-0.1,-0.1)}]Aset.north west) {A};
        \node[anchor=north east, font=\large\bfseries] at ([shift={(-0.7,-0.1)}]Jset.north west) {J};
        \node[anchor=north east, font=\large\bfseries] at ([shift={(-1.0,-0.1)}]Sset.north west) {\(S_O\)};
        \node[anchor=north east, font=\large\bfseries]
            at ([shift={(-2.8,-0.2)}]Cset.north west)
            {\(\Cn_{\mathrm{rev}}(S_O)\)};
        \node[anchor=north east, font=\large\bfseries] at ([shift={(-0.7,-0.1)}]Hset.north west) {\(\hat{\mathbb{S}}\)};

\node[font=\scriptsize, align=left, anchor=north west]
  at ([xshift=2mm,yshift=0]Jset.south west)
  {\(\forall j\in J,\ \forall \hat{s}\in\hat{\mathbb{S}}:\)};
\node[font=\scriptsize, align=left, anchor=north west]
  at ([xshift=1mm,yshift=-4mm]Jset.south west)
  {\(\ d_{\Cn}(j,\hat{s})=0\)};

        \begin{scope}[on background layer]
            \node[fit=(Sset), fill=gray!10, rounded corners, inner sep=0pt] {};
            \node[fit=(Hset), fill=gray!15, rounded corners, inner sep=0pt] {};
        \end{scope}

        \draw[relDer] (a1.west) -- (j1.east);
        \draw[relDer] (a2.west) -- (j1.east);
        \draw[relDer] (a2.west) -- (j2.east);
        \draw[relDer] (a3.west) -- (j2.east);
        \draw[relDer] (a3.west) -- (j3.east);
        \draw[relDer] (a4.west) -- (j3.east);

        \draw[relDer] (j1.west) -- (c1.east);
        \draw[relDer] (j2.west) -- (c1.east);
        \draw[relDer] (j2.west) -- (c2.east);
        \draw[relDer] (j3.west) -- (c2.east);
        \draw[relDer] (j2.west) -- (c4.east);
        \draw[relDer] (j3.west) -- (c4.east);

        \draw[relDer] (a1.west) -- (c3.east);
        \draw[relDer] (a3.west) -- (c3.east);

        \draw[relGen] (a1.west) -- (lh01.east);
        \draw[relGen] (j2.west) -- (lh01.east);
        \draw[relGen] (a2.west) -- (lh02.east);
        \draw[relGen] (j2.west) -- (lh03.east);
        \draw[relGen] (a4.west) -- (lh03.east);
        \draw[relGen] (a1.west) -- (lhp1.east);
        \draw[relGen] (j3.west) -- (lhp2.east);

    \end{scope}

    \path let \p1 = (Cset.north), \p2 = (Cset.south), \p3 = (Cset.east) in
      node [draw, thick, rounded corners, align=center,
            minimum width=6.8cm,
            anchor=west,
            at={($(\x3 + 2mm, 0.5*\y1 + 0.5*\y2)$)},
            minimum height=\y1-\y2] (Rset) {};
    \node[anchor=north, font=\large\bfseries] at ($(Rset.north)+(0,-0.15)$) {Hypergraph \(\Gamma_0\)};

    \path (Rset.center) ++(135:2.6cm) +(0, 0.8) node[elemA] (ha1) {\(a_1\)};
    \path (Rset.center) ++(45:2.6cm) +(0, 0.8) node[elemA] (ha2) {\(a_2\)};
    \path (Rset.center) ++(-45:2.6cm) +(0, 0.8) node[elemA] (ha3) {\(a_3\)};    
    \path (Rset.center) ++(-135:2.6cm) +(0, 0.8) node[elemA] (ha4) {\(a_4\)};    

    \draw[fill=gray!20, opacity=0.5, draw=black, line width=0.8pt]
        (ha1.east) -- (ha2.west) -- (ha3.west) -- cycle;
    \node[font=\small] at (barycentric cs:ha1=2.5,ha2=2.5,ha3=1) {\(W_{\hat s_1}\)};
    \node[anchor=north, font=\bfseries] at ($(ha2.north)+(0,0.5)$) {\(\{a_2\}=W_{\hat s_2}\)};
    \draw[fill=gray!40, opacity=0.5, draw=black, line width=0.8pt]
        (ha2.west) -- (ha3.west) -- (ha4.east) -- cycle;
    \node[font=\small] at (barycentric cs:ha2=1,ha3=2.5,ha4=2.5) {\(W_{\hat s_3}\)};
    \node[anchor=north, font=\bfseries] at ($(ha1.north)+(0,0.5)$) {\(\{a_1\}=W_{\hat s_4}\)};
    \draw[thick, black] (ha3.west) -- (ha4.east) node[midway, sloped, below, font=\small] {\(W_{\hat s_5}\)};    

    \node[anchor=south east, font=\scriptsize, align=right]
        at ([xshift=-2mm,yshift=26mm]Rset.south east)
        {\(\displaystyle W_{\hat s}\coloneqq \{a\in A:\ \hat s\in R_0(a)\}\)};

    \node[anchor=north, font=\bfseries] (piA)
        at ($(ha4.south)+(0.8,-1.1)$)
        {\(\pi_A(a)\coloneqq P_O(a\mid A)\)};

    \node[anchor=north, font=\bfseries] (PA)
        at ($(ha4.south)+(-0.2,-2.3)$)
        {\(P_A\coloneqq P_O(A)\)};

    \node[anchor=north, font=\bfseries] (HG)
        at ($(ha4.south)+(4.58,-1.1)$)
        {\(H_{\Gamma_0}(\pi_A)\)};

    \node[anchor=north, font=\bfseries] (RD)
        at ($(ha4.south)+(3.8,-2.3)$)
        {\(R_{\sem}(0)= P_A\,H_{\Gamma_0}(\pi_A)\)};

    \tikzset{dashedArc/.style={dashed, draw=black, thick}}

    \draw[dashedArc]
         (pia1.east) .. controls +(2.4, -0.3) and ($(piA.west)+(-0.8, 0)$) .. (piA.west);       
    \draw[dashedArc]
        (pia2.east) .. controls +(2.0, -0.3) and ($(piA.west)+(-1.0, 0)$) .. (piA.west);
    \draw[dashedArc]
        (pia3.east) .. controls +(1.0, -0.3) and ($(piA.west)+(-1.2, 0)$) .. (piA.west);
    \draw[dashedArc]
        (pia4.east) .. controls +(0.8, -0.1) and ($(piA.west)+(-1.0, 0)$) .. (piA.west);

    \draw[dashedArc]
        (pa.east) .. controls +(2.4, 0) and ($(PA.west)+(-1, 0)$) .. (PA.west);

    \draw[dashedArc]
        (piA.east) .. controls +(0.9,0) and +(-0.9,0) .. (HG.west);
    \draw[dashedArc]
        (HG.south) .. controls +(0,-0.5) and +(0,0.5) .. ($(RD.north)+(0.8,0)$);
    \draw[dashedArc]
        (PA.east) .. controls +(0,0) .. (RD.west);
    
    \begin{scope}[shift={($(Cset.south)-(0,1.2)$)}, anchor=north]
       
        \node[elemA] (lega) at (-4.5,0.7) {\(\scriptstyle a\)};
        \node[anchor=west,font=\small] at (-4.25,0.55) {atomic basis element};       
        \node[elemJ] (legj) at (0.25,0.7) {};
        \node[anchor=west,font=\small] at (0.55,0.55) {redundant element};
        \node[elemC] (legc) at (4.8,0.7) {};
        \node[anchor=west,font=\small] at (5.1,0.55) {closure-added element};

        \node[elemH] (legh) at (-4.5,0.2) {\(\scriptstyle \hat{s}\)};
        \node[anchor=west,font=\small] at (-4.25,0.05) {reconstruction element};
        \draw[relDer] (-0.05,0.05) -- (0.65,0.05);
        \node[anchor=west,font=\small] at (0.62,0.05) {logical derivation};
        \draw[relGen] (4.5,0.05) -- (5.2,0.05);
        \node[anchor=west,font=\small] at (5.2,0.05) {reconstruction generation};

    \end{scope}

\end{tikzpicture}

\caption{Two-view visualization of closure-preserving semantic RD at zero distortion.
\emph{Left (core/closure view).} The observed log \(S_O\) decomposes canonically as \(S_O=A\uplus J\),
where the irredundant core \(A=\Atom_{\mathrm{rev}}(S_O)\) preserves closure
\(\Cn_{\mathrm{rev}}(A)=\Cn_{\mathrm{rev}}(S_O)\), while the redundant part \(J\) is distortion-invisible under admissibility
(\(\forall j\in J,\forall \hat s\in\hat{\mathbb{S}}:\ d_{\Cn}(j,\hat s)=0\)).
Dashed arrows indicate exemplar zero-distortion witnesses \(\hat s\in R_0(s)\) for single-symbol replacement.
The labels \(P_A=P_O(A)\) and \(\pi_A(a)=P_O(a\mid A)\) highlight the core probability mass and the core-conditioned distribution.
\emph{Right (confusability hypergraph view).} Each reconstruction symbol \(\hat s\in\hat{\mathbb{S}}\) induces a witness hyperedge
\(W_{\hat s}\coloneqq \{a\in A:\ \hat s\in R_0(a)\}\), and such hyperedges encode which core symbols are mutually confusable at
zero distortion. The confusability hypergraph \(\Gamma_0\) (Definition~\ref{def:gamma0}) consists of all nonempty
\(W\subseteq A\) with \(\bigcap_{a\in W}R_0(a)\neq\emptyset\) (hence it is downward closed); the figure depicts representative
hyperedges \(W_{\hat s}\) witnessing these overlaps.
The resulting hypergraph entropy yields the exact perfect-rollback rate \(R_{\sem}(0)=P_A\,H_{\Gamma_0}(\pi_A)\).}
\label{fig:sem-rd-zero}
\end{figure}

\begin{definition}[Semantic rate--distortion function]
\label{def:rsem}
Let \(\hat{S}\) be generated from \(S\) via a test channel \(P_{\hat{S}\mid S}\) with \(\hat{S}\in\hat{\mathbb{S}}\).
Define
\[
R_{\sem}(D) \coloneqq
\inf_{P_{\hat{S}\mid S}: \E[d_{\Cn}(S,\hat{S})]\le D} I(S;\hat{S}).
\]
\end{definition}

We interpret \(R_{\sem}(D)\) as the single-letter rate--distortion function of the memoryless extension in which
\((S_1,\dots,S_n)\) are drawn i.i.d.\ from \(P_O\), and the encoder/decoder operates on blocks of length \(n\) with average distortion constraint.
Under this standard setting, \(R_{\sem}(D)\) gives the optimal asymptotic rate; see, e.g., \cite{cover2006elements,csizsar2011information}.

The next theorem gives the exact zero-distortion rate in terms of a confusability hypergraph induced by the closure-preserving zero-distortion sets.

\begin{theorem}[Zero-distortion hypergraph entropy law]
\label{thm:zero}
Under Assumptions~\ref{assm:source}--\ref{assm:admissible}, the minimum achievable rate for zero distortion is
\[
R_{\sem}(0)=P_A \, H_{\Gamma_0}(\pi_A),
\]
with the convention \(P_A H_{\Gamma_0}(\pi_A)=0\) when \(P_A=0\), and
\(R_{\sem}(0)=+\infty\) if \(P_A>0\) and the feasible set in the definition of
\(H_{\Gamma_0}(\pi_A)\) is empty (i.e., no zero-distortion channel exists).
\end{theorem}

\begin{proof}
If \(P_A=0\), then \(S\in J\) almost surely.  Choose any fixed admissible symbol
\(\hat s_0\in\hat{\mathbb S}\) (non‑empty by Assumption~\ref{assm:admissible}) and
put \(\hat S\equiv \hat s_0\) deterministically.  Proposition~\ref{prop:redundant-zero}
gives \(d_{\Cn}(S,\hat S)=0\) a.s.\ and \(I(S;\hat S)=0\); hence \(R_{\sem}(0)=0\).

Assume \(P_A>0\). Let \(\mathsf{B}\coloneqq \1[S\in A]\) be the core indicator.
Consider any test channel \(P_{\hat{S}\mid S}\) achieving \(\E[d_{\Cn}(S,\hat{S})]=0\).
Then \(d_{\Cn}(S,\hat{S})=0\) almost surely, hence \(\hat{S}\in R_0(S)\) almost surely.

\emph{Converse.}
By the chain rule and non-negativity of mutual information,
\[
I(S;\hat{S}) \ge I(S;\hat{S}\mid \mathsf{B})
\ge P_A\, I(S;\hat{S}\mid \mathsf{B}=1).
\]
Conditioned on \(\mathsf{B}=1\), the random variable \(S\) has distribution \(\pi_A\), i.e., \(S\mid(\mathsf{B}=1)\equiv A^\star\).
Thus
\[
I(S;\hat{S}) \ge P_A\, I(A^\star;\hat{S}).
\]
To expose the induced confusability structure, define the random subset
\[
W \coloneqq \{a\in A : \hat{S}\in R_0(a)\}.
\]
Then \(A^\star\in W\) almost surely (since \(\hat{S}\in R_0(A^\star)\) a.s.), and \(W\in\Gamma_0\) almost surely by construction:
indeed, \(\hat{S}\in \bigcap_{a\in W} R_0(a)\), so the intersection is nonempty.
Moreover, \(W\) is a measurable function of \(\hat{S}\), hence by data processing,
\(I(A^\star;\hat{S}) \ge I(A^\star;W)\).
Therefore,
\[
I(S;\hat{S}) \ge P_A\, I(A^\star;W) \ge P_A\, H_{\Gamma_0}(\pi_A).
\]
Taking the infimum over all zero-distortion channels yields \(R_{\sem}(0)\ge P_A H_{\Gamma_0}(\pi_A)\).

\emph{Achievability.}
If the feasible set for \(H_{\Gamma_0}(\pi_A)\) is empty, then by definition
\(H_{\Gamma_0}(\pi_A)=+\infty\) and the inequality \(R_{\sem}(0)\le P_A H_{\Gamma_0}(\pi_A)\)
is trivial.  Otherwise, let \(P_{W\mid A^\star}\) be \(\epsilon\)-optimal in Definition~\ref{def:gamma0}, so that
\(I(A^\star;W)\le H_{\Gamma_0}(\pi_A)+\epsilon\), with \(A^\star\in W\) and \(W\in\Gamma_0\) almost surely.
Since \(\mathbb{S}\) (hence \(\hat{\mathbb{S}}\)) is finite, for each \(W\in\Gamma_0\) we can fix a selector
\(\psi(W)\in \bigcap_{a\in W}R_0(a)\) (choose any element in the nonempty intersection).
Define a test channel \(P_{\hat{S}\mid S}\) as follows:
\begin{itemize}
  \item If \(S=a\in A\), draw \(W\sim P_{W\mid A^\star}(\cdot\mid a)\) and output \(\hat{S}=\psi(W)\).
  \item If \(S=j\in J\), output \(\hat{S}\) according to the marginal distribution of \(\psi(W)\) under \(A^\star\sim\pi_A\) and \(W\mid A^\star\).
\end{itemize}
By construction, when \(S\in A\) we have \(\hat{S}\in \bigcap_{a\in W}R_0(a)\subseteq R_0(S)\), hence zero distortion.
When \(S\in J\), Proposition~\ref{prop:redundant-zero} ensures zero distortion for any admissible output, hence the channel achieves
\(\E[d_{\Cn}(S,\hat{S})]=0\).

It remains to bound the rate. By construction, \(\hat{S}=\psi(W)\) is a (deterministic) function of \(W\), and hence
\[
I(A^\star;\hat{S}) \le I(A^\star;W)
\]
by the data processing inequality (Markov chain \(A^\star \to W \to \hat{S}\)).

Moreover, the output distribution under \(S\in J\) is chosen to match the unconditional output marginal induced by
\(A^\star\sim\pi_A\) and \(W\mid A^\star\). Therefore \(\hat{S}\) has the same distribution under \(\mathsf{B}=0\) and \(\mathsf{B}=1\),
and thus \(I(\mathsf{B};\hat{S})=0\). Using the chain rule,
\[
I(S;\hat{S}) = I(\mathsf{B};\hat{S}) + I(S;\hat{S}\mid \mathsf{B})
= P_A\, I(A^\star;\hat{S})
\le P_A\, I(A^\star;W)
\le P_A\,(H_{\Gamma_0}(\pi_A)+\epsilon).
\]
Since \(\epsilon>0\) is arbitrary, \(R_{\sem}(0)\le P_A H_{\Gamma_0}(\pi_A)\).
\end{proof}

Figure~\ref{fig:sem-rd-zero} visualizes the key structural invariants: the partition \(S_O=A\uplus J\) with
\(\Cn_{\mathrm{rev}}(A)=\Cn_{\mathrm{rev}}(S_O)\), and the way overlaps of \(R_0(a)\) over \(a\in A\) induce \(\Gamma_0\) and hence \(R_{\sem}(0)\).

\begin{corollary}[Disjoint-core formula]
\label{cor:disjoint}
If the core is \emph{disjoint}, i.e., \(R_0(a_1)\cap R_0(a_2)=\emptyset\) for all distinct \(a_1,a_2\in A\), and \(A\subseteq \hat{\mathbb{S}}\),
then
\[
R_{\sem}(0)=P_A H(\pi_A).
\]
\end{corollary}

\begin{proof}
Under disjointness, the hypergraph \(\Gamma_0\) has only singleton hyperedges, hence \(H_{\Gamma_0}(\pi_A)=H(\pi_A)\).
Substitute into Theorem~\ref{thm:zero}.
\end{proof}

\subsection{General distortion: core-only factorization of the semantic rate--distortion function}
\label{sec:sem-general}

Under admissibility, Proposition~\ref{prop:redundant-zero} implies \(d_{\Cn}(S,\hat S)=0\) almost surely on \(\{S\in J\}\),
and Lemma~\ref{lem:dist-basic} gives \(0\le d_{\Cn}\le 1\). Therefore every test channel satisfies the uniform upper bound
\[
\E[d_{\Cn}(S,\hat S)]
=\Pr[S\in A]\cdot \E[d_{\Cn}(S,\hat S)\mid S\in A]
\le P_A.
\]
In this sense, distortion budgets \(D>P_A\) are vacuous, and in particular \(R_{\sem}(D)=0\) for all \(D\ge P_A\).

Separately, the standard ``zero-rate'' distortion threshold is
\[
D_0 \;\coloneqq\; \min_{\hat s\in\hat{\mathbb{S}}} \E[d_{\Cn}(S,\hat s)],
\]
corresponding to always outputting a fixed reconstruction symbol. By the bound above, \(D_0\le P_A\), and \(R_{\sem}(D)=0\) for all
\(D\ge D_0\).

For convenience, let \(\mathsf{B}\coloneqq \1[S\in A]\) denote the indicator of whether the source symbol lies in the core.

\begin{theorem}[Core factorization for all \(D\)]
\label{thm:factorization}
Assume Assumption~\ref{assm:admissible}. Let \(P_A\coloneqq P_O(A)\).
Then for all \(D\ge 0\),
\[
R_{\sem}(D)=
\begin{cases}
0, & P_A=0,\\[1mm]
P_A \cdot R_A\!\bigl(\min\{D/P_A,\,1\}\bigr), & P_A>0,
\end{cases}
\]
where \(R_A(\delta)\) is the rate--distortion function of the \emph{core sub-source}
\(A^\star\sim \pi_A(\cdot)=P_O(\cdot\mid A)\) under the \emph{same} single-symbol closure distortion \(d_{\Cn}\)
(Definition~\ref{def:distortion}, with the same reference log \(S_O\)), namely
\[
R_A(\delta)\coloneqq 
\inf_{P_{\hat{S}\mid A^\star}:\ \E[d_{\Cn}(A^\star,\hat{S})]\le \delta} I(A^\star;\hat{S}).
\]
In particular, since \(d_{\Cn}\in[0,1]\), we have \(R_{\sem}(D)=0\) for all \(D\ge P_A\).
\end{theorem}

\begin{proof}
The argument is the same core-decomposition reduction as in \cite{xu2026rate}, specialized to the reversible closure \(\Cn_{\mathrm{rev}}\).

If \(P_A=0\), then \(S\in J\) almost surely and, by Proposition~\ref{prop:redundant-zero}, every admissible reconstruction achieves
\(\E[d_{\Cn}(S,\hat S)]=0\). Hence \(R_{\sem}(D)=0\) for all \(D\ge 0\).

Assume \(P_A>0\) and let \(\mathsf{B}\coloneqq \1[S\in A]\).

\medskip
\noindent\textbf{Step 1 (distortion depends only on the core).}
By Proposition~\ref{prop:redundant-zero}, \(d_{\Cn}(S,\hat S)=0\) almost surely on \(\{\mathsf{B}=0\}\) for any test channel.
Therefore,
\[
\E[d_{\Cn}(S,\hat S)]
= \E[d_{\Cn}(S,\hat S)\,\mathbf{1}[\mathsf{B}=1]]
= P_A\,\E[d_{\Cn}(A^\star,\hat S)],
\]
where \(A^\star\sim\pi_A\) denotes \(S\) conditioned on \(\{\mathsf{B}=1\}\).
In particular, any constraint \(\E[d_{\Cn}(S,\hat S)]\le D\) implies
\(\E[d_{\Cn}(A^\star,\hat S)]\le \min\{D/P_A,1\}\).

\medskip
\noindent\textbf{Step 2 (converse).}
By the chain rule and non-negativity,
\[
I(S;\hat S)=I(\mathsf{B};\hat S)+I(S;\hat S\mid \mathsf{B})
\ge P_A\,I(S;\hat S\mid \mathsf{B}=1)
= P_A\,I(A^\star;\hat S).
\]
The induced conditional channel \(P_{\hat S\mid A^\star}\) satisfies the core distortion constraint from Step~1,
hence \(I(A^\star;\hat S) \ge R_A(\min\{D/P_A,1\})\). Therefore
\(I(S;\hat S)\ge P_A R_A(\min\{D/P_A,1\})\), and taking the infimum yields the desired lower bound.

\medskip
\noindent\textbf{Step 3 (achievability).}
Let \(P^{\star}_{\hat S\mid A^\star}\) be optimal (or \(\epsilon\)-optimal) for \(R_A(\min\{D/P_A,1\})\),
and let \(P^{\star}_{\hat S}\) be its induced output marginal under \(A^\star\sim\pi_A\).
Define a channel on \(S\) by
\[
P_{\hat S\mid S}(\cdot\mid s)\coloneqq
\begin{cases}
P^{\star}_{\hat S\mid A^\star}(\cdot\mid s), & s\in A,\\
P^{\star}_{\hat S}(\cdot), & s\in J.
\end{cases}
\]
It is admissible since \(\hat S\in\hat{\mathbb{S}}\) by construction, and it achieves
\(\E[d_{\Cn}(S,\hat S)]=P_A\,\E[d_{\Cn}(A^\star,\hat S)]\le D\).
Moreover, \(\hat S\) has the same marginal distribution under \(\mathsf{B}=0\) and \(\mathsf{B}=1\), so \(I(\mathsf{B};\hat S)=0\) and
\[
I(S;\hat S)=P_A\,I(A^\star;\hat S)=P_A\,R_A(\min\{D/P_A,1\}) \ (\text{up to }\epsilon).
\]
Letting \(\epsilon\to 0\) yields the upper bound and completes the proof.
\end{proof}

\section{Rollback Safety via Semantic Rate--Distortion}
\label{sec:rb}

Section~\ref{sec:zero} gives a semantic RD theory under closure-preserving fidelity.
Here we connect it to rollback computing in two steps:
(i) we show that controlling the semantic distortion \(d_{\Cn}\) already yields exact and computable rollback-safety guarantees;
(ii) we then introduce a rollback-task loss \(\ell_{\mathrm{rb}}\) that provides tighter, task-aligned guarantees and a refined RD tradeoff.

\subsection{Rollback safety induced by closure-preserving fidelity}
\label{sec:rb-from-closure}

We work with the same single-symbol edit model as in Section~\ref{sec:zero}.

\begin{definition}[Edited log and violation indicators]
\label{def:viol-indic}
Fix the reference log \(S_O\subseteq \mathbb{S}\). For \(s\in S_O\) and \(\hat s\in \hat{\mathbb{S}}\), define the edited log
\[
S_O[s\leftarrow \hat s]\ \coloneqq\ (S_O\setminus\{s\})\cup\{\hat s\}.
\]
Define the \emph{rollback violation indicator}
\[
\mathsf{Viol}(s,\hat s)\ \coloneqq\ \1\!\big[\mathsf{RB}(S_O[s\leftarrow \hat s])\neq \mathsf{RB}(S_O)\big],
\]
and the \emph{closure violation indicator}
\[
\mathsf{Viol}_{\Cn}(s,\hat s)\ \coloneqq\ \1\!\big[\Cn_{\mathrm{rev}}(S_O[s\leftarrow \hat s])\neq \Cn_{\mathrm{rev}}(S_O)\big],
\]
where \(\mathsf{RB}(S)=\Cn_{\mathrm{rev}}(S)\cap \mathbb{Q}_{\mathrm{rb}}\) (Definition~\ref{def:rb-capability}).
\end{definition}

\paragraph{Remark (set-based edit semantics).}
Logs are modeled as \emph{sets} of facts. Hence the edit operation
\(S_O[s\leftarrow \hat s]=(S_O\setminus\{s\})\cup\{\hat s\}\) should be read as ``delete then insert'':
if \(\hat s\in S_O\setminus\{s\}\), the operation degenerates to deleting \(s\).
When \(S,\hat S\) are random variables, the notation \(S_O[S\leftarrow \hat S]\) is interpreted pointwise.

\begin{lemma}[Closure invariance implies rollback invariance]
\label{lem:closure-implies-rb}
For all \(s\in S_O\) and \(\hat s\in \hat{\mathbb{S}}\),
\[
\mathsf{Viol}(s,\hat s)\ \le\ \mathsf{Viol}_{\Cn}(s,\hat s).
\]
Consequently, for any test channel \(P_{\hat S\mid S}\),
\[
\Pr[\mathsf{Viol}(S,\hat S)=1]\ \le\ \Pr[\mathsf{Viol}_{\Cn}(S,\hat S)=1].
\]
\end{lemma}
\begin{proof}
If \(\Cn_{\mathrm{rev}}(S_O[s\leftarrow \hat s])=\Cn_{\mathrm{rev}}(S_O)\), then intersecting both sides with \(\mathbb Q_{\mathrm{rb}}\)
gives \(\mathsf{RB}(S_O[s\leftarrow \hat s])=\mathsf{RB}(S_O)\) (Proposition~\ref{prop:rb-vs-closure}), hence \(\mathsf{Viol}(s,\hat s)=0\).
\end{proof}

\begin{lemma}[A distortion-to-closure-violation bound under admissibility]
\label{lem:dCn-to-violCn}
Assume Assumption~\ref{assm:admissible}. For all \(s\in S_O\) and \(\hat s\in \hat{\mathbb{S}}\),
\[
\mathsf{Viol}_{\Cn}(s,\hat s)\ \le\ |\Cn_{\mathrm{rev}}(S_O)|\cdot d_{\Cn}(s,\hat s).
\]
Consequently, for any test channel \(P_{\hat S\mid S}\) over \(\hat{\mathbb S}\),
\[
\Pr[\mathsf{Viol}_{\Cn}(S,\hat S)=1]\ \le\ |\Cn_{\mathrm{rev}}(S_O)|\cdot \E[d_{\Cn}(S,\hat S)].
\]
\end{lemma}

\begin{proof}
If \(\mathsf{Viol}_{\Cn}(s,\hat s)=0\) then the inequality is trivial. Otherwise let
\[
U \ \coloneqq\ \Cn_{\mathrm{rev}}(S_O)\ \cup\ \Cn_{\mathrm{rev}}(S_O[s\leftarrow \hat s]).
\]
By admissibility, \(\hat s\in \Cn_{\mathrm{rev}}(S_O)\), and clearly \(S_O\setminus\{s\}\subseteq S_O\subseteq \Cn_{\mathrm{rev}}(S_O)\).
Hence
\[
S_O[s\leftarrow \hat s]=(S_O\setminus\{s\})\cup\{\hat s\}\ \subseteq\ \Cn_{\mathrm{rev}}(S_O),
\]
and by monotonicity and idempotence,
\[
\Cn_{\mathrm{rev}}(S_O[s\leftarrow \hat s])\ \subseteq\ \Cn_{\mathrm{rev}}(\Cn_{\mathrm{rev}}(S_O))\ =\ \Cn_{\mathrm{rev}}(S_O).
\]
Therefore \(U=\Cn_{\mathrm{rev}}(S_O)\), so \(1\le |U|=|\Cn_{\mathrm{rev}}(S_O)|\).

Since the two closures are not equal, their intersection is a strict subset of the union, hence
\(|\Cn_{\mathrm{rev}}(S_O)\cap \Cn_{\mathrm{rev}}(S_O[s\leftarrow \hat s])|\le |U|-1\).
By Definition~\ref{def:fidelity},
\[
\mathsf F_{\Cn}\bigl(S_O,\ S_O[s\leftarrow \hat s]\bigr)
=\frac{|\cap|}{|U|}
\le \frac{|U|-1}{|U|}
=1-\frac{1}{|U|},
\]
so \(d_{\Cn}(s,\hat s)=1-\mathsf F_{\Cn}\ge 1/|U|=1/|\Cn_{\mathrm{rev}}(S_O)|\).
Thus \(\1[\mathsf{Viol}_{\Cn}(s,\hat s)=1]\le |\Cn_{\mathrm{rev}}(S_O)|\cdot d_{\Cn}(s,\hat s)\), and taking expectations yields the
probability bound.
\end{proof}

\begin{theorem}[Perfect rollback-observable preservation via zero semantic distortion]
\label{thm:perfect-rb-observable-from-Rsem}
Under Assumptions~\ref{assm:source}--\ref{assm:admissible}, any scheme that achieves
zero semantic distortion \(\E[d_{\Cn}(S,\hat S)]=0\) also achieves
\[
\Pr[\mathsf{Viol}(S,\hat S)=1]=0,
\]
where \(\mathsf{Viol}\) is the indicator of a change in the rollback observable \(\mathsf{RB}(\cdot)\).
Moreover, the minimum rate among all such zero-distortion schemes is exactly
\(R_{\sem}(0)=P_A\,H_{\Gamma_0}(\pi_A)\).
If Assumption~\ref{assm:rb-adequacy} holds, then zero semantic distortion is \emph{sufficient}
for perfect preservation of the rollback judgments represented by \(\mathbb{Q}_{\mathrm{rb}}\),
and the minimal rate needed when enforcing this sufficient condition is \(R_{\sem}(0)\).
\end{theorem}

\begin{proof}
If \(\E[d_{\Cn}(S,\hat S)]=0\) then \(d_{\Cn}(S,\hat S)=0\) almost surely, hence \(\Cn_{\mathrm{rev}}(S_O[S\leftarrow \hat S])=\Cn_{\mathrm{rev}}(S_O)\)
almost surely by Lemma~\ref{lem:dist-basic}. Lemma~\ref{lem:closure-implies-rb} then gives \(\mathsf{Viol}(S,\hat S)=0\) almost surely.
The rate identity is by Definition~\ref{def:rsem} and Theorem~\ref{thm:zero}.
\end{proof}

\begin{corollary}[A conservative approximate rollback-safety design rule]
\label{cor:approx-rb-safety-from-Rsem}
Assume Assumption~\ref{assm:admissible}. For any test channel with \(\E[d_{\Cn}(S,\hat S)]\le D\),
\[
\Pr[\mathsf{Viol}(S,\hat S)=1]\ \le\ |\Cn_{\mathrm{rev}}(S_O)|\cdot D.
\]
In particular, a sufficient condition for \(\Pr[\mathsf{Viol}(S,\hat S)=1]\le \varepsilon\) is
\(D\le \varepsilon/|\Cn_{\mathrm{rev}}(S_O)|\), and the corresponding semantic rate bound is
\[
R \ \ge\ R_{\sem}(\varepsilon/|\Cn_{\mathrm{rev}}(S_O)|)
\ =\ P_A\cdot R_A\!\Bigl(\min\Bigl\{\frac{\varepsilon}{|\Cn_{\mathrm{rev}}(S_O)|P_A},1\Bigr\}\Bigr),
\]
where the last identity uses Theorem~\ref{thm:factorization}.
\end{corollary}

\begin{proof}
Combine Lemma~\ref{lem:closure-implies-rb} and Lemma~\ref{lem:dCn-to-violCn}.
The rate expression is Theorem~\ref{thm:factorization} applied at \(D=\varepsilon/|\Cn_{\mathrm{rev}}(S_O)|\).
\end{proof}

\subsection{Rollback-task loss and rollback-task rate--distortion}
\label{sec:rb-rd}

Closure-preserving fidelity yields conservative guarantees. We now introduce a task-aligned loss that measures \emph{how many}
rollback-relevant judgments change.

\begin{definition}[Rollback-task loss: ``how many rollback judgments are wrong?'']
\label{def:rb-loss}
Fix the reference log \(S_O\subseteq\mathbb{S}\).
For \(s\in S_O\) and \(\hat s\in \hat{\mathbb{S}}\), define
\[
\ell_{\mathrm{rb}}(s,\hat s)
\ \coloneqq\
\bigl|\mathsf{RB}(S_O[s\leftarrow \hat s])\ \Delta\ \mathsf{RB}(S_O)\bigr|,
\]
where \(\Delta\) denotes symmetric difference.
\end{definition}

\begin{lemma}[Range and a tight safety bridge]
\label{lem:rb-loss-range}
For all \(s\in S_O\) and \(\hat s\in\hat{\mathbb{S}}\),
\[
0\le \ell_{\mathrm{rb}}(s,\hat s)\le |\mathbb Q_{\mathrm{rb}}|,
\qquad
\mathsf{Viol}(s,\hat s)\ \le\ \1[\ell_{\mathrm{rb}}(s,\hat s)\ge 1]\ \le\ \ell_{\mathrm{rb}}(s,\hat s).
\]
Consequently, for any test channel \(P_{\hat S\mid S}\),
\[
\Pr[\mathsf{Viol}(S,\hat S)=1]\ \le\ \E[\ell_{\mathrm{rb}}(S,\hat S)].
\]
\end{lemma}
\begin{proof}
The range bound holds because \(\mathsf{RB}(\cdot)\subseteq\mathbb Q_{\mathrm{rb}}\).
If \(\mathsf{Viol}(s,\hat s)=1\) then the symmetric difference is nonempty, so \(\ell_{\mathrm{rb}}(s,\hat s)\ge 1\).
Since \(\ell_{\mathrm{rb}}\) is a nonnegative integer, \(\1[\ell_{\mathrm{rb}}\ge 1]\le \ell_{\mathrm{rb}}\).
Taking expectations gives the final inequality.
\end{proof}

\begin{definition}[Rollback-task rate--distortion function]
\label{def:Rrb}
Let \(\hat S\) be generated from \(S\sim P_O\) via a test channel \(P_{\hat S\mid S}\) over \(\hat{\mathbb{S}}\).
Define
\[
R_{\mathrm{rb}}(L)
\ \coloneqq\
\inf_{P_{\hat S\mid S}:\ \E[\ell_{\mathrm{rb}}(S,\hat S)]\le L}\ I(S;\hat S).
\]
\end{definition}

As in Section~\ref{sec:zero}, admissibility (\(\hat{\mathbb{S}}\subseteq \Cn_{\mathrm{rev}}(S_O)\), Assumption~\ref{assm:admissible})
eliminates the redundant part \(J\) from the information-theoretic tradeoff.

\begin{lemma}[Redundant facts are rollback-task invisible under admissibility]
\label{lem:rb-loss-redundant}
Assume Assumption~\ref{assm:admissible}. For any \(j\in J\) and any \(\hat s\in\hat{\mathbb{S}}\),
\[
\ell_{\mathrm{rb}}(j,\hat s)=0.
\]
\end{lemma}
\begin{proof}
By Proposition~\ref{prop:redundant-zero}, for \(j\in J\) and any admissible \(\hat s\),
\(\Cn_{\mathrm{rev}}(S_O[j\leftarrow \hat s])=\Cn_{\mathrm{rev}}(S_O)\).
Intersect with \(\mathbb Q_{\mathrm{rb}}\) to obtain \(\mathsf{RB}(S_O[j\leftarrow \hat s])=\mathsf{RB}(S_O)\), hence \(\ell_{\mathrm{rb}}(j,\hat s)=0\).
\end{proof}

\begin{definition}[Core rollback-task RD function]
\label{def:Rrb-core}
Let \(P_A\coloneqq P_O(A)\). When \(P_A>0\), let \(\pi_A(\cdot)=P_O(\cdot\mid A)\) and \(A^\star\sim \pi_A\).
Define
\[
R_{\mathrm{rb},A}(\lambda)
\ \coloneqq\
\inf_{P_{\hat S\mid A^\star}:\ \E[\ell_{\mathrm{rb}}(A^\star,\hat S)]\le \lambda}\ I(A^\star;\hat S),
\]
where \(\ell_{\mathrm{rb}}(a,\hat s)\) is evaluated with the same reference log \(S_O\) via Definition~\ref{def:rb-loss}.
\end{definition}

\begin{theorem}[Core-only factorization of \(R_{\mathrm{rb}}(L)\)]
\label{thm:Rrb-factorization}
Assume Assumption~\ref{assm:admissible}. Let \(P_A\coloneqq P_O(A)\).
Then for all \(L\ge 0\),
\[
R_{\mathrm{rb}}(L)=
\begin{cases}
0, & P_A=0,\\[1mm]
P_A\cdot R_{\mathrm{rb},A}\!\bigl(\min\{L/P_A,\ |\mathbb Q_{\mathrm{rb}}|\}\bigr), & P_A>0.
\end{cases}
\]
In particular, \(R_{\mathrm{rb}}(L)=0\) for all \(L\ge P_A|\mathbb Q_{\mathrm{rb}}|\).
\end{theorem}

\begin{proof}
If \(P_A=0\), then \(S\in J\) almost surely and \(\ell_{\mathrm{rb}}(S,\hat S)=0\) for any admissible \(\hat S\) by
Lemma~\ref{lem:rb-loss-redundant}, so \(R_{\mathrm{rb}}(L)=0\).

Assume \(P_A>0\) and let \(\mathsf B\coloneqq \1[S\in A]\). By Lemma~\ref{lem:rb-loss-redundant},
\(\ell_{\mathrm{rb}}(S,\hat S)=0\) almost surely on \(\{\mathsf B=0\}\), hence
\[
\E[\ell_{\mathrm{rb}}(S,\hat S)]
= P_A\,\E[\ell_{\mathrm{rb}}(A^\star,\hat S)].
\]
Moreover,
\[
I(S;\hat S)=I(\mathsf B;\hat S)+I(S;\hat S\mid \mathsf B)\ge P_A\,I(A^\star;\hat S).
\]
For achievability, let \(P_{\hat S\mid A^\star}^*\) be an optimal (or \(\epsilon\)-optimal)
core test channel for \(R_{\mathrm{rb},A}\bigl(\min\{L/P_A,|\mathbb Q_{\mathrm{rb}}|\}\bigr)\),
and let \(P_{\hat S}^*\) be its output marginal under \(\pi_A\).  Define the overall channel by
\[
P_{\hat S\mid S}(\cdot\mid s)=
\begin{cases}
P_{\hat S\mid A^\star}^*(\cdot\mid s), & s\in A,\\
P_{\hat S}^*(\cdot), & s\in J .
\end{cases}
\]
Because \(\hat S\) has the same marginal on \(\mathsf B=0\) and \(\mathsf B=1\), we have
\(I(\mathsf B;\hat S)=0\) and consequently
\(I(S;\hat S)=P_A\,I(A^\star;\hat S)\le P_A\bigl(R_{\mathrm{rb},A}(\cdots)+\epsilon\bigr)\).
The expected loss becomes \(P_A\,\E[\ell_{\mathrm{rb}}(A^\star,\hat S)]\le L\).
Letting \(\epsilon\to0\) completes the proof.

\end{proof}

\subsection{Perfect rollback: hypergraph law and comparison with perfect closure preservation}
\label{sec:rb-zero}

\begin{definition}[Zero-loss reconstruction sets for rollback]
\label{def:rb-zero-set}
For \(s\in S_O\), define
\[
R^{\mathrm{rb}}_0(s)\coloneqq \{\hat s\in\hat{\mathbb{S}}:\ \ell_{\mathrm{rb}}(s,\hat s)=0\}.
\]
\end{definition}

\begin{definition}[Rollback confusability hypergraph and hypergraph entropy]
\label{def:gamma0-rb}
Define the rollback confusability hypergraph \(\Gamma^{\mathrm{rb}}_0\) on vertex set \(A\) by
\[
\Gamma^{\mathrm{rb}}_0 \coloneqq \Bigl\{ W\subseteq A:\ W\neq\emptyset,\ \bigcap_{a\in W}R^{\mathrm{rb}}_0(a)\neq\emptyset \Bigr\}.
\]
Define its hypergraph entropy
\[
H_{\Gamma^{\mathrm{rb}}_0}(\pi_A)\coloneqq 
\inf_{\substack{P_{W\mid A^\star}:\\ W\in\Gamma^{\mathrm{rb}}_0\ \text{a.s.},\ A^\star\in W\ \text{a.s.}}}
I(A^\star;W).
\]
\end{definition}

\begin{theorem}[Perfect-rollback hypergraph entropy law]
\label{thm:rb-zero-hypergraph}
Under Assumptions~\ref{assm:source}--\ref{assm:admissible}, the minimum achievable rate for
\emph{perfect rollback} \(\E[\ell_{\mathrm{rb}}(S,\hat S)]=0\) is
\[
R_{\mathrm{rb}}(0)=P_A\,H_{\Gamma^{\mathrm{rb}}_0}(\pi_A),
\]
with the convention \(P_A H_{\Gamma^{\mathrm{rb}}_0}(\pi_A)=0\) when \(P_A=0\), and
\(R_{\mathrm{rb}}(0)=+\infty\) if \(P_A>0\) and the feasible set in the definition of
\(H_{\Gamma^{\mathrm{rb}}_0}(\pi_A)\) is empty (i.e., no zero-loss test channel exists).
\end{theorem}

\begin{proof}
The proof follows the same hypergraph-entropy converse/achievability pattern as Theorem~\ref{thm:zero},
with \(R^{\mathrm{rb}}_0(\cdot)\) in place of \(R_0(\cdot)\) and \(\ell_{\mathrm{rb}}\) in place of \(d_{\Cn}\).
Redundant symbols contribute neither loss nor rate under admissibility (Lemma~\ref{lem:rb-loss-redundant}),
so the optimization reduces to the core and yields the factor \(P_A\).
\end{proof}

\begin{lemma}[Monotonicity of hypergraph entropy under hyperedge enlargement]
\label{lem:hypergraph-entropy-monotone}
Let \(\Gamma\subseteq \Gamma'\) be two families of nonempty subsets of \(A\). Then
\[
H_{\Gamma'}(\pi_A)\ \le\ H_{\Gamma}(\pi_A).
\]
\end{lemma}
\begin{proof}
In the definition of \(H_{\Gamma}(\pi_A)\), the infimum is taken over channels \(P_{W\mid A^\star}\) with \(W\in\Gamma\) a.s.
Replacing \(\Gamma\) by the larger family \(\Gamma'\) enlarges the feasible set, hence the infimum cannot increase.
\end{proof}

\begin{proposition}[Perfect rollback is no harder than perfect closure preservation]
\label{prop:rb-vs-sem-zero}
Under Assumptions~\ref{assm:source}--\ref{assm:admissible},
\[
R_{\mathrm{rb}}(0)\le R_{\sem}(0).
\]
Moreover, if rollback-equivalence coincides with full closure equivalence on the admissible reconstruction class
(e.g., when \(\mathbb Q_{\mathrm{rb}}=\mathbb{S}\)), then \(R_{\mathrm{rb}}(0)=R_{\sem}(0)\).
\end{proposition}

\begin{proof}
For every \(s\in S_O\), closure preservation implies rollback preservation (Proposition~\ref{prop:rb-vs-closure}), hence
\(R_0(s)\subseteq R^{\mathrm{rb}}_0(s)\). Therefore \(\Gamma_0\subseteq \Gamma^{\mathrm{rb}}_0\), and
Lemma~\ref{lem:hypergraph-entropy-monotone} yields \(H_{\Gamma^{\mathrm{rb}}_0}(\pi_A)\le H_{\Gamma_0}(\pi_A)\).
Multiply by \(P_A\) and apply Theorem~\ref{thm:zero} and Theorem~\ref{thm:rb-zero-hypergraph}.
The equality condition follows from the stated equivalence.
\end{proof}

\paragraph{Implication for computable design loops.}
The finite-alphabet pipeline of Section~\ref{sec:rcn-ba-theory} applies to both objectives:
one may compute \(R_A(\delta)\) using \(d_{\Cn}\), or compute \(R_{\mathrm{rb},A}(\lambda)\) using \(\ell_{\mathrm{rb}}\),
via the same Blahut--Arimoto machinery after constructing the corresponding distortion/loss matrix.

\section{RCN/rPES Instantiation: Discipline-Indexed Closures and Core Structure}
\label{sec:rcn-theory}

This section instantiates closure-preserving semantic rate--distortion on reversible causal nets (RCN),
and outlines the parallel event-structure (rPES) view via the RCN--rPES correspondence of \cite{melgratti2024reversible}.
Our goal here is \emph{theoretical}: define a discipline-indexed family of monotone closures capturing rollback-relevant semantics,
and derive the resulting closure cores and zero-distortion confusability structure.
Section~\ref{sec:rcn-exp} then provides a full numerical evaluation for the
resulting RD-guided log compression schemes.

For concreteness, the formal rules and core-structure proofs below are presented in full for the RCN instantiation;
the rPES side is stated in a parallel (abbreviated) form in Section~\ref{sec:rpes-instantiation} and follows by the same monotone-closure
arguments and/or transport along the correspondence.

\subsection{RCN/rPES background, logging interface, and discipline-indexed closures}
\label{sec:rcn-background-interface}

We adopt the rcn setting of \cite[\S7]{melgratti2024reversible}.
Let \(\overline{T}\) be the finite set of forward transitions and \(\underline{T}\) the set of backward transitions,
with pairing \(\mathsf{pair}:\underline{T}\to\overline{T}\).
Forward causality is induced by inhibitor arcs:
\[
t\ \ell\ t' \quad\Longleftrightarrow\quad {^\bullet t}\cap {^\circ t'}\neq\emptyset,
\]
and forward conflict is induced by shared preset:
\[
t\ \#\ t' \quad\Longleftrightarrow\quad {^\bullet t}\cap {^\bullet t'}\neq\emptyset.
\]
As in Sections~\ref{sec:prelim}--\ref{sec:rb}, an execution state is logged as a finite fact base.

A (forward) configuration is a finite set \(X\subseteq \overline{T}\) of forward transitions currently present.
The run-specific fact base is
\[
S_O^{\mathrm{rcn}}(X)\ \coloneqq\ \{\mathsf{In}(t): t\in X\}.
\]
Static net structure facts (preset/inhibitor/pairing) live in a background theory \(\mathcal{B}_{\mathrm{rcn}}\) \cite[\S7]{melgratti2024reversible}.
Concretely, \(\mathcal{B}_{\mathrm{rcn}}\) contains ground facts for \(\mathsf{Pre}(p,t)\), \(\mathsf{Inhib}(p,t)\), and the pairing map
(e.g., a predicate \(\mathsf{Pair}(\underline{t},t)\) encoding \(\mathsf{pair}(\underline{t})=t\)),
so that \(\Cn_{\mathrm{rev}}(S)\) reasons over \(\mathcal{B}_{\mathrm{rcn}}\cup S\) rather than requiring these structural facts to be logged per run.

A central message of \cite{melgratti2024reversible} is that different reversing disciplines
(causal, cause-respecting, inverse causal, out-of-causal-order) change which facts are rollback-relevant.
To capture this systematically within our closure-preserving interface, we define a
\emph{discipline-indexed} family of monotone Horn systems
\(\mathsf{PS}^{\mathrm{disc}}_{\mathrm{rcn}}\) inducing closures \(\Cn^{\mathrm{disc}}_{\mathrm{rcn}}\).

\subsubsection{Common structural rules: causality and conflict evidence}
\label{sec:rcn-common-rules}

We use the following shared predicates in \(\Sigma_{\mathrm{rev}}\):
\[
\mathsf{Pre}(p,t),\ \mathsf{Inhib}(p,t),\ \mathsf{In}(t),\
\mathsf{Cause}(t,t'),\ \mathsf{Cause}^\star(t,t'),\
\mathsf{Conf}(t,t').
\]
All ground \(\mathsf{Pre}(\cdot,\cdot)\) and \(\mathsf{Inhib}(\cdot,\cdot)\) facts are assumed to be part of the background theory
\(\mathcal{B}_{\mathrm{rcn}}\).
We write \(\mathsf{Conf}(\cdot,\cdot)\) as a shorthand for the conflict predicate \(\mathsf{Conflict}(\cdot,\cdot)\) introduced in
Section~\ref{sec:prelim}.
Rules (R1)--(R2) provide a fact-level encoding of the structural relations \(t\ \ell\ t'\) and \(t\ \#\ t'\) introduced above:
\(\mathsf{Cause}(t,t')\) captures inhibitor-induced causality, and \(\mathsf{Conf}(t,t')\) captures shared-preset conflict.

\paragraph{(R1) Inhibitor-induced causality and its transitive closure.}
\[
\mathsf{Pre}(p,t)\wedge \mathsf{Inhib}(p,t')\Rightarrow \mathsf{Cause}(t,t'),
\qquad
\mathsf{Cause}(t,t')\Rightarrow \mathsf{Cause}^\star(t,t'),
\]
\[
\mathsf{Cause}^\star(t,u)\wedge \mathsf{Cause}^\star(u,t')\Rightarrow \mathsf{Cause}^\star(t,t').
\]

\paragraph{(R2) Conflict evidence (structural).}
\[
\mathsf{Pre}(p,t)\wedge \mathsf{Pre}(p,t')\wedge (t\neq t')\Rightarrow \mathsf{Conf}(t,t').
\]
We treat \(t\neq t'\) as a built-in inequality predicate on constants (standard in Datalog-style presentations).

\paragraph{Remark (inequality).}
When implementing the closure via a Datalog engine, the side condition \(t\neq t'\) can be treated either as a built-in interpreted predicate
over constants, or compiled away by precomputing extensional \(\mathsf{Neq}(t,t')\) facts for all distinct transition identifiers.

\subsubsection{Causal-consistent (``causal'') discipline closure}
\label{sec:rcn-causal-closure}

\paragraph{(CC1) Semantic downward closure of presence.}
\[
\mathsf{In}(t)\wedge \mathsf{Cause}^\star(u,t)\Rightarrow \mathsf{In}(u).
\]

\paragraph{(CC2) Non-maximality and rollback blockers.}
\[
\mathsf{In}(t)\wedge \mathsf{Cause}^\star(t,u)\wedge \mathsf{In}(u)\Rightarrow \mathsf{NonMax}(t),
\qquad
\mathsf{NonMax}(t)\Rightarrow \mathsf{Blocked}(\underline{t}).
\]
Optionally record order constraints:
\[
\mathsf{In}(t)\wedge \mathsf{In}(u)\wedge \mathsf{Cause}^\star(t,u)\Rightarrow
\mathsf{MustUndoBefore}(\underline{t},\underline{u}).
\]

\paragraph{Closure and rollback query set.}
Let \(\mathsf{PS}^{\mathrm{causal}}_{\mathrm{rcn}}\) be the Horn system comprising (R1)--(R2) plus (CC1)--(CC2).
Define
\[
\Cn^{\mathrm{causal}}_{\mathrm{rcn}}(S)\ \coloneqq\ \{s\in\mathbb{S}:\ \mathcal{B}_{\mathrm{rcn}}\cup S\vdash_{\mathsf{PS}^{\mathrm{causal}}_{\mathrm{rcn}}} s\},
\quad
\mathbb{Q}^{\mathrm{causal}}_{\mathrm{rb}}\coloneqq
\{\mathsf{Blocked}(\underline{t}),\mathsf{NonMax}(t),\mathsf{MustUndoBefore}(\underline{t},\underline{u})\}.
\]

\subsubsection{Cause-respecting discipline closure}
\label{sec:rcn-cause-respecting-closure}

We introduce \(\mathsf{Sust}(t,u)\) to represent the \emph{sustained causation} relation \(\ll\) described in \cite{melgratti2024reversible}.
Unlike the basic inhibitor-induced causality \(\mathsf{Cause}^\star\) (Rule~R1), sustained causation encodes a more persistent dependency between events that is relevant for cause-respecting reversibility.
For the purposes of our formal development, it is sufficient to treat \(\mathsf{Sust}\) as either
(i)~a pre-computed background predicate in \(\mathcal{B}_{\mathrm{rcn}}\), or
(ii)~a derived predicate defined by a fixed (and finite) set of Horn rules.

\paragraph{(CR1) Sustained-nonmaximality and blockers.}
\[
\mathsf{In}(t)\wedge \mathsf{Sust}(t,u)\wedge \mathsf{In}(u)\Rightarrow \mathsf{NonMax}^{\ll}(t),
\qquad
\mathsf{NonMax}^{\ll}(t)\Rightarrow \mathsf{Blocked}(\underline{t}).
\]

\paragraph{Closure and query set.}
Let \(\mathsf{PS}^{\mathrm{cr}}_{\mathrm{rcn}}\) be (R1)--(R2) plus (CC1) plus (CR1).
Define
\[
\Cn^{\mathrm{cr}}_{\mathrm{rcn}}(S)\coloneqq \{s:\mathcal{B}_{\mathrm{rcn}}\cup S\vdash_{\mathsf{PS}^{\mathrm{cr}}_{\mathrm{rcn}}} s\},
\quad
\mathbb{Q}^{\mathrm{cr}}_{\mathrm{rb}}\coloneqq \{\mathsf{Blocked}(\underline{t}),\mathsf{NonMax}^{\ll}(t)\}.
\]

\subsubsection{Inverse causal discipline closure}
\label{sec:rcn-inverse-causal-closure}

\paragraph{Why (CC1) is not assumed under inverse-causal reversibility.}
In causal-consistent settings, forward presence is downward closed: if a consequence is present, then its causes are present as well,
which is precisely what (CC1) captures in a monotone proof system.
In contrast, inverse-causal (and more generally out-of-causal-order) reversibility may admit states in which a cause has been undone
while some of its former consequences remain present.
Therefore downward closure of \(\mathsf{In}(\cdot)\) is not a discipline-invariant semantic law, and we do not include (CC1) in
\(\mathsf{PS}^{\mathrm{inv}}_{\mathrm{rcn}}\).
This is exactly the mechanism by which inverse-causal rollback evidence can force additional log atoms into the closure core
(Section~\ref{sec:rcn-core-by-discipline}).

Inverse causal reversibility changes rollback gating: undoing is prevented by the presence of \emph{causes}
rather than \emph{consequences} \cite{melgratti2024reversible}.

\paragraph{(IC1) Blockers from causes present.}
\[
\mathsf{In}(u)\wedge \mathsf{Cause}^\star(u,t)\Rightarrow \mathsf{CausePresent}(u,t),
\qquad
\mathsf{CausePresent}(u,t)\Rightarrow \mathsf{Blocked}(\underline{t}).
\]

\paragraph{Closure and query set.}
Let \(\mathsf{PS}^{\mathrm{inv}}_{\mathrm{rcn}}\) be (R1)--(R2) plus (IC1).
Define
\[
\Cn^{\mathrm{inv}}_{\mathrm{rcn}}(S)\coloneqq \{s:\mathcal{B}_{\mathrm{rcn}}\cup S\vdash_{\mathsf{PS}^{\mathrm{inv}}_{\mathrm{rcn}}} s\},
\quad
\mathbb{Q}^{\mathrm{inv}}_{\mathrm{rb}}\coloneqq \{\mathsf{Blocked}(\underline{t}),\mathsf{CausePresent}(u,t)\}.
\]

\subsubsection{rPES instantiation: an event-structure overview}
\label{sec:rpes-instantiation}

To complement the detailed RCN development, we briefly outline how the same discipline-indexed, closure-based approach applies to the setting of reversible prime event structures (rPES).
An rPES provides a finite set of forward events \(\mathcal{E}\), reverse events \(\underline{\mathcal{E}}\), and structural relations such as causality and conflict \cite{melgratti2024reversible}.
A run state (configuration) is represented by a finite set \(X\subseteq \mathcal{E}\) of currently present forward events, for which we log the fact base
\[
S_O^{\mathrm{pes}}(X)\ \coloneqq\ \{\mathsf{In}(e): e\in X\}.
\]
All structural relations needed by rollback reasoning (e.g., \(\mathsf{Cause}(\cdot,\cdot)\), \(\mathsf{Conflict}(\cdot,\cdot)\))
are treated as background facts in a finite theory \(\mathcal{B}_{\mathrm{pes}}\), and we assume that
\(\mathsf{Cause}^\star\) is available either as a precomputed background relation or as the closure of a fixed Horn rule set
defining transitive closure. Thus, closures are computed over \(\mathcal{B}_{\mathrm{pes}}\cup S\).

\paragraph{Discipline-indexed monotone closures and rollback queries.}
Using the same monotone ``positive evidence'' design as in Section~\ref{sec:prelim}, we define rPES closures by Horn rules
syntactically parallel to the RCN ones, yielding operators \(\Cn^{\mathrm{disc}}_{\mathrm{pes}}\) for
\(\mathrm{disc}\in\{\mathrm{causal},\mathrm{cr},\mathrm{inv}\}\).
Rollback query sets \(\mathbb Q^{\mathrm{disc}}_{\mathrm{rb}}\) are defined analogously to the RCN case, with transitions replaced by events
(e.g., \(\mathsf{Blocked}(\underline{e})\), \(\mathsf{NonMax}(e)\), \(\mathsf{MustUndoBefore}(\underline{e},\underline{e}')\), etc.).

\paragraph{Status of the rPES instantiation.}
In this version, the rPES discussion should be read as a correspondence-level companion to the fully formal RCN development.
All theorem-level proofs in this section are given for RCN; a full transport theorem for cores, distortion matrices, and rollback observables on rPES is deferred to a subsequent version.

\subsection{Core computation depends on the rollback discipline}
\label{sec:rcn-core-by-discipline}

The closure-based core \(A=\Atom_{\mathrm{rev}}(S_O)\) (Definition~\ref{def:core}) is discipline-dependent through the chosen
\(\Cn_{\mathrm{rev}}\).
We now characterize the causal and cause-respecting cores as frontier-shaped, and give an explicit
counterexample for inverse causal semantics.

\subsubsection{Causal and Cause-respecting disciplines: frontier-core results}
\label{sec:rcn-frontier-core}

\begin{definition}[Frontier (maximal set)]
Let \(X\subseteq\overline{T}\). Define the frontier of \(X\) (w.r.t.\ the derived causal reachability \(\mathsf{Cause}^\star\)) as
\[
\Fr(X)\ \coloneqq\ \{t\in X:\ \nexists u\in X\ \text{with}\ \mathsf{Cause}^\star(t,u)\ \text{and}\ t\neq u\}.
\]
\end{definition}

\begin{assumption}[Acyclicity of causal reachability on configurations]
\label{assm:rcn-acyclic}
For every reachable configuration \(X\subseteq\overline{T}\), the restriction of
\(\mathsf{Cause}^\star\) to \(X\) is acyclic.
\end{assumption}

\begin{proposition}[Frontier core for disciplines with downward-closed presence]
\label{prop:frontier-core-downward}
Fix a configuration \(X\subseteq\overline{T}\), let
\[
S_O=S_O^{\mathrm{rcn}}(X), \qquad F\coloneqq \Fr(X), \qquad S_F\coloneqq S_O^{\mathrm{rcn}}(F).
\]
Assume Assumption~\ref{assm:rcn-acyclic}, that rule \emph{(CC1)} belongs to the closure system,
and that no rule other than \emph{(CC1)} derives atoms of the form \(\mathsf{In}(\cdot)\).
Then
\[
\Atom_{\mathrm{rev}}(S_O)=S_F.
\]
In particular, the identity holds for both \(\Cn^{\mathrm{causal}}_{\mathrm{rcn}}\) and
\(\Cn^{\mathrm{cr}}_{\mathrm{rcn}}\).
\end{proposition}

\begin{proof}
We first show that
\[
\Cn_{\mathrm{rev}}(S_F)=\Cn_{\mathrm{rev}}(S_O).
\]
Since \(S_F\subseteq S_O\), monotonicity gives
\[
\Cn_{\mathrm{rev}}(S_F)\subseteq \Cn_{\mathrm{rev}}(S_O).
\]
For the reverse inclusion, let \(t\in X\).

If \(t\in F\), then \(\mathsf{In}(t)\in S_F\subseteq \Cn_{\mathrm{rev}}(S_F)\).

If \(t\notin F\), then \(t\) has a strict \(\mathsf{Cause}^\star\)-successor in \(X\).
Because \(X\) is finite and \(\mathsf{Cause}^\star\!\upharpoonright_X\) is acyclic, repeated choice of a strict
\(\mathsf{Cause}^\star\)-successor must terminate at a \(\mathsf{Cause}^\star\)-maximal element \(f\in F\) such that
\(\mathsf{Cause}^\star(t,f)\).
Since \(\mathsf{In}(f)\in S_F\), rule \emph{(CC1)} yields
\[
\mathsf{In}(f)\wedge \mathsf{Cause}^\star(t,f)\Rightarrow \mathsf{In}(t),
\]
so \(\mathsf{In}(t)\in \Cn_{\mathrm{rev}}(S_F)\).

Hence \(S_O\subseteq \Cn_{\mathrm{rev}}(S_F)\), and monotonicity plus idempotence imply
\[
\Cn_{\mathrm{rev}}(S_O)\subseteq \Cn_{\mathrm{rev}}(\Cn_{\mathrm{rev}}(S_F))=\Cn_{\mathrm{rev}}(S_F).
\]
Therefore \(\Cn_{\mathrm{rev}}(S_F)=\Cn_{\mathrm{rev}}(S_O)\).

Next we show that each frontier atom is irredundant.
Fix \(f\in F\) and suppose, towards a contradiction, that
\[
\mathsf{In}(f)\in \Cn_{\mathrm{rev}}\bigl(S_F\setminus\{\mathsf{In}(f)\}\bigr).
\]
By assumption, the only rule that derives new \(\mathsf{In}(\cdot)\) atoms is \emph{(CC1)}.
Consequently, any derivation of \(\mathsf{In}(f)\) from \(S_F\setminus\{\mathsf{In}(f)\}\) must
involve an instance of (CC1) with conclusion \(\mathsf{In}(f)\).  The premises of (CC1) require
some \(\mathsf{In}(g)\) together with \(\mathsf{Cause}^\star(f,g)\).  Since the derivation
uses only facts from \(S_F\setminus\{\mathsf{In}(f)\}\) and Horn rules, the atom
\(\mathsf{In}(g)\) must eventually be supported by a ground fact in
\(S_F\setminus\{\mathsf{In}(f)\}\); that is, there exists
\(\mathsf{In}(g)\in S_F\setminus\{\mathsf{In}(f)\}\) such that
\(\mathsf{Cause}^\star(f,g)\).  But then \(g\in X\) is a strict
\(\mathsf{Cause}^\star\)-successor of \(f\), contradicting the maximality of \(f\in F\).
Hence
\[
\mathsf{In}(f)\notin \Cn_{\mathrm{rev}}\bigl(S_F\setminus\{\mathsf{In}(f)\}\bigr)
\qquad\text{for all } f\in F.
\]

Finally, consider the deletion scan of Definition~\ref{def:core}.
By the previous paragraph, no frontier atom is ever deleted.
When a non-frontier atom \(\mathsf{In}(t)\) is scanned, the current working set still contains \(S_F\), and from the first part
of the proof we already know
\[
\mathsf{In}(t)\in \Cn_{\mathrm{rev}}(S_F).
\]
Thus
\[
\mathsf{In}(t)\in \Cn_{\mathrm{rev}}(A\setminus\{\mathsf{In}(t)\}),
\]
so \(\mathsf{In}(t)\) is deleted by the scan.
Therefore the scan keeps exactly the frontier atoms, and the final core is \(S_F\).
\end{proof}

\begin{corollary}[Frontier is the core under \(\Cn^{\mathrm{causal}}_{\mathrm{rcn}}\)]
\label{cor:frontier-core-causal}
Fix \(X\subseteq\overline{T}\) and \(S_O=S_O^{\mathrm{rcn}}(X)\).
Under \(\Cn^{\mathrm{causal}}_{\mathrm{rcn}}\) and Assumption~\ref{assm:rcn-acyclic},
\[
\Atom_{\mathrm{rev}}(S_O)=S_O^{\mathrm{rcn}}(\Fr(X)).
\]
\end{corollary}

\begin{corollary}[Frontier is the core under \(\Cn^{\mathrm{cr}}_{\mathrm{rcn}}\)]
\label{cor:frontier-core-cr}
Fix \(X\subseteq\overline{T}\) and \(S_O=S_O^{\mathrm{rcn}}(X)\).
Under \(\Cn^{\mathrm{cr}}_{\mathrm{rcn}}\) and Assumption~\ref{assm:rcn-acyclic},
\[
\Atom_{\mathrm{rev}}(S_O)=S_O^{\mathrm{rcn}}(\Fr(X)).
\]
\end{corollary}

Figure~\ref{fig:rcn-fig8-core-frontier-closure} visualizes, on the Fig.~8 instance of \cite{melgratti2024reversible}, how the closure interface
\(\Cn_{\mathrm{rev}}\) makes redundancy invisible and reduces logging complexity to the frontier-shaped core.
Starting from \(S_O=\{\mathsf{In}(a),\mathsf{In}(b),\mathsf{In}(c),\mathsf{In}(d)\}\), the deletion scan (Definition~\ref{def:core})
deletes \(\mathsf{In}(a)\) and \(\mathsf{In}(b)\) because, in that instance, they are derivable from \(\{\mathsf{In}(c),\mathsf{In}(d)\}\) under the
downward-closure rule (CC1), while \(\mathsf{In}(c)\) and \(\mathsf{In}(d)\) are irredundant.
Thus the core coincides with the frontier: \(A=S_O^{\mathrm{rcn}}(\Fr(X))\).
The bottom row of the figure also previews why inverse-causal closures can force additional causes into the core
to preserve blocker evidence (cf.\ Proposition~\ref{prop:inverse-core-not-frontier}).

\begin{figure}[t]
\centering
\begin{tikzpicture}[
  font=\scriptsize,
  >=Latex,
  place/.style={circle, draw=black, thick, minimum size=3.5mm, inner sep=0pt},
  trans/.style={rectangle, draw=black, thick, minimum width=3.5mm, minimum height=3.5mm, inner sep=1pt},
  flow/.style={->, thick, draw=black},
  flowg/.style={->, thick, draw=black!45}, 
  inh/.style={thick, draw=red!80!black},   
  inhmark/.style={circle, draw=red!80!black, fill=white, inner sep=1.2pt},
  token/.style={circle, fill=black, inner sep=1.3pt},
  frontierBox/.style={draw=red!85!black, dashed, very thick, rounded corners, inner sep=7.5pt},
  redundantBox/.style={draw=black!55, dashed, very thick, rounded corners, inner sep=7.5pt},
  panel/.style={draw=black, rounded corners, thick, inner sep=6pt, align=left},
  tinyPanel/.style={draw=black, rounded corners, thick, inner sep=4pt, align=left},
  node distance=6mm
]

\begin{scope}[x=1cm,y=1cm, scale=0.95, local bounding box=LEFT]
  \node[inner sep=0, minimum size=0] at (-0.2,4.3) {};

  \begin{scope}[yshift=-0.2cm]
    \draw[flow] (-0.1,3.85) -- +(0.7,0) node[anchor=west] {flow arc};
    \draw[inh]  (2.0,3.85) -- +(0.7,0) node[anchor=west] {inhibitor arc};
    \node[inhmark] at (2.65,3.85) {};
    \node[token] at (4.8,3.85) {};
    \node[anchor=west] at (4.9,3.85) {initial token};
  \end{scope}

  \begin{scope}[yshift=-0.5cm]
    \node[place] (s0) at (3.0, 3.4) {};
    \node[anchor=north, font=\small\bfseries] (s0m) at ($(s0.south)+(0,0.05)$) {$s_0$};
    \node[place] (s1) at (0.0, 2.4) {};
    \node[anchor=north, font=\small\bfseries] (s1m) at ($(s1.north)+(0,0.35)$) {$s_1$};
    \node[place] (s2) at (1.5, 2.4) {};
    \node[anchor=north, font=\small\bfseries] (s2m) at ($(s2.north)+(0,0.35)$) {$s_2$};
    \node[place] (s3) at (3.0, 2.4) {};
    \node[anchor=north, font=\small\bfseries] (s3m) at ($(s3.north)+(0,0.35)$) {$s_3$};
    \node[place] (s4) at (4.5, 2.4) {};
    \node[anchor=north, font=\small\bfseries] (s4m) at ($(s4.north)+(0,0.35)$) {$s_4$};
    \node[place] (s5) at (6.6, 2.4) {};
    \node[anchor=north, font=\small\bfseries] (s5m) at ($(s5.north)+(0,0.35)$) {$s_5$};

    \node[place] (s6) at (0.0, 0.4) {$s_6$};
    \node[place] (s7) at (2.2, 0.4) {$s_7$};
    \node[place] (s8) at (4.1, 0.4) {$s_8$};
    \node[place] (s9) at (6.6, 0.4) {$s_9$};

    \node[trans] (a) at (0.0, 1.4) {$a$};
    \node[trans] (b) at (2.2, 1.4) {$b$};
    \node[trans] (c) at (4.1, 1.4) {$c$};
    \node[trans] (d) at (6.6, 1.4) {$d$};

    \node[trans] (ub) at (3.0, -0.2) {$\underline{b}$};
    \node[trans] (uc) at (5.2, -0.2) {$\underline{c}$};

    \node[token] at (s0.center) {};
    \node[token] at (s1.center) {};
    \node[token] at (s2.center) {};
    \node[token] at (s3.center) {};
    \node[token] at (s4.center) {};
    \node[token] at (s5.center) {};

    \draw[flow] (s1) -- (a);
    \draw[flow] (s2) -- (a.north east);    
    \draw[flow] (s0) to[out=200, in=90, looseness=1.3] (a.north east);
    \draw[flow] (a) -- (s6);

    \draw[flow] (s2) -- (b);
    \draw[flow] (s3) -- (b);
    \draw[flow] (b) -- (s7);

    \draw[flow] (s0) to[out=340, in=120, looseness=1.2] (c.north);
    \draw[flow] (s4) -- (c);
    \draw[flow] (c) -- (s8);

    \draw[flow] (s5) -- (d);
    \draw[flow] (d) -- (s9);

    \draw[flowg] (s7) -- (ub);
    \draw[flowg] (ub) to[out=190, in=280, looseness=1.1] (s2.south);
    \draw[flowg] (ub.north east) to[out=0, in=290, looseness=1.1] (s3.south east);

    \draw[flowg] (s8.south) -- (uc.west);
    \draw[flowg] (uc.north) to[out=0, in=0, looseness=0.4] (s4.south east);
    \draw[flowg] (uc.north) to[out=20, in=30, looseness=0.9] (s0.east);

    \inhArc{(s3.east)}{(c.north west)};    
    \inhArc{(s3.south)}{(ub.north)};
    \inhArc{(s4.south)}{(uc.north west)};
    \inhArc{(s5.south)}{(uc.east)};
    \inhArc{(s8.south)}{(ub.east)};

  \end{scope}
\end{scope}

\begin{scope}[x=1cm,y=1cm, scale=0.95]
    \node[frontierBox, 
          inner ysep=2.74mm, inner xsep=2.04mm,  
          fit=(c)(d),
          label={[red!85!black, yshift=-2.8cm]above:\textbf{Frontier }$\{c,d\}$}] (FBOX) {};
    \node[redundantBox, 
          inner ysep=2.74mm, inner xsep=2.04mm,
          fit=(a)(b),
          label={[black!55, yshift=-2.8cm]above:\textbf{Redundant }$\{a,b\}$}] (RBOX) {};
\end{scope}

\path let \p3 = ($(LEFT.east)-(LEFT.west)$) in \pgfextra{\xdef\Lwidth{\x3}};

\node[panel, text width=7.4cm,
      anchor=north west,
      at=(LEFT.north east),
      xshift=0.4cm
     ] (RIGHT) {
    \textbf{Closure, deletion scan, and resulting core (fact base).}\\[2pt]
    \(S_O=\{\mathsf{In}(a),\mathsf{In}(b),\mathsf{In}(c),\mathsf{In}(d)\}\).\\[3pt]
    \textbf{Sample derived facts in \(\Cn_{\mathrm{rev}}(S_O)\):}\\
    \(\mathsf{NonMax}(b)\), \(\mathsf{Blocked}(\underline{b})\),
    \(\mathsf{MustUndoBefore}(\underline{b},\underline{c})\) (illustrative).\\[4pt]
    \textbf{Deletion scan (order \(a\prec b\prec c\prec d\); closure \(\Cn^{\mathrm{causal}}_{\mathrm{rcn}}\)):}\\[-2pt]
    \[
    \renewcommand{\arraystretch}{1.5}
    \begin{array}{c|c|c}
      \text{scan } \mathsf{In}(\cdot) & \text{decision} & \text{reason (informal)}\\\hline
      \mathsf{In}(a) & \text{delete} & \mathsf{In}(a)\in \Cn^{\mathrm{causal}}_{\mathrm{rcn}}(\{\mathsf{In}(c),\mathsf{In}(d)\})\\
      \mathsf{In}(b) & \text{delete} & \mathsf{In}(b)\in \Cn^{\mathrm{causal}}_{\mathrm{rcn}}(\{\mathsf{In}(c),\mathsf{In}(d)\})\\
      \mathsf{In}(c) & \text{keep}   & \mathsf{In}(c)\notin \Cn^{\mathrm{causal}}_{\mathrm{rcn}}(\{\mathsf{In}(d)\})\\
      \mathsf{In}(d) & \text{keep}   & \mathsf{In}(d)\notin \Cn^{\mathrm{causal}}_{\mathrm{rcn}}(\{\mathsf{In}(c)\})
    \end{array}
    \]
    \textbf{Core and frontier coincide (causal/cause-respecting):}\\
    \(A=\{\mathsf{In}(c),\mathsf{In}(d)\}\), corresponding to frontier \(\{c,d\}\).
};

\path let \p1 = ($(LEFT.north)-(LEFT.south)$) in 
  \pgfextra{\xdef\Lheight{\y1}};
\path let \p2 = ($(RIGHT.north)-(RIGHT.south)$) in 
  \pgfextra{\xdef\Rheight{\y2}};

\begin{scope}[on background layer]
  \node[draw=black, thick, rounded corners, inner sep=3mm,
        anchor=north, minimum height=\Rheight,
        minimum width=\Lwidth+6mm,
        at=(LEFT.north)] (LEFTBOX) {};
\end{scope}

\path let \p2 = ($(RIGHT.south east)-(LEFTBOX.south west)$) in
  node[tinyPanel, anchor=north west,
       text width=\x2-8pt,
       align=left
      ] (B) at ($(LEFTBOX.south west)+(0,-0.1cm)$) {%
    \textbf{Inverse-causal discipline (contrast).}
    Under \(\Cn^{\mathrm{inv}}_{\mathrm{rcn}}\), blocker evidence depends on \emph{causes being present}
    (e.g., \(\mathsf{In}(u)\wedge \mathsf{Cause}^\star(u,t)\Rightarrow \mathsf{CausePresent}(u,t)\Rightarrow \mathsf{Blocked}(\underline{t})\)),
    so the core generally must include additional presence facts such as \(\mathsf{In}(a)\) and \(\mathsf{In}(b)\) to preserve
    \(\mathsf{CausePresent}(\cdot,\cdot)\) / \(\mathsf{Blocked}(\cdot)\) consequences.
};
\end{tikzpicture}

\caption{RCN instance based on \cite[Fig.~8]{melgratti2024reversible} illustrating the relationship between reversible semantic closure
\(\Cn_{\mathrm{rev}}\), deletion-scan core \(A\), and frontier \(\Fr(X)\).
Left: the net with frontier \(\{c,d\}\) (red dashed box) and redundant events \(\{a,b\}\) (gray dashed box).
Right: fact base, example closure consequences, deletion scan outcome, and \(A=\{\mathsf{In}(c),\mathsf{In}(d)\}\) matching \(\Fr(X)\)
under causal/cause-respecting closures; inverse-causal reversal can enlarge the core by requiring causes to preserve blocker evidence:
for instance, under \(\Cn^{\mathrm{inv}}_{\mathrm{rcn}}\) one may need \(\mathsf{In}(a)\) to derive
\(\mathsf{CausePresent}(a,b)\) and hence \(\mathsf{Blocked}(\underline{b})\); omitting \(\mathsf{In}(a)\) can therefore change the rollback closure.}
\label{fig:rcn-fig8-core-frontier-closure}
\end{figure}

\subsubsection{Inverse causal discipline: frontier can be insufficient}
\label{sec:rcn-core-inverse-example}

\begin{proposition}[Inverse causal core need not be the frontier]
\label{prop:inverse-core-not-frontier}
There exists a net and a configuration \(X\) such that the core under \(\Cn^{\mathrm{inv}}_{\mathrm{rcn}}\) is \emph{not}
\(S_O^{\mathrm{rcn}}(\Fr(X))\).
\end{proposition}

\begin{proof}
Consider a chain \(a<b\) with \(X=\{a,b\}\) and no conflicts. Then \(\Fr(X)=\{b\}\).
Under (IC1), \(\mathsf{In}(a)\wedge \mathsf{Cause}^\star(a,b)\Rightarrow \mathsf{CausePresent}(a,b)\Rightarrow \mathsf{Blocked}(\underline{b})\).
If we log only \(\{\mathsf{In}(b)\}\), then \(\mathsf{In}(a)\) is missing and the blocker evidence disappears, changing rollback closure.
Thus \(\mathsf{In}(a)\) is core-relevant; the frontier alone is insufficient.
Formally, let \(S_O=\{\mathsf{In}(a),\mathsf{In}(b)\}\).  Under \(\Cn^{\mathrm{inv}}_{\mathrm{rcn}}\),
we have \(\Cn^{\mathrm{inv}}_{\mathrm{rcn}}(\{\mathsf{In}(b)\}) = \{\mathsf{In}(b)\} \cup
\text{background consequences}\), which does not contain \(\mathsf{In}(a)\);
hence \(\mathsf{In}(a)\notin\Cn^{\mathrm{inv}}_{\mathrm{rcn}}(S_O\setminus\{\mathsf{In}(a)\})\).
Therefore the deletion scan keeps \(\mathsf{In}(a)\), and the core contains
\(\{\mathsf{In}(a),\mathsf{In}(b)\}\), whereas \(\Fr(X)=\{b\}\).
Thus the core is strictly larger than the frontier.
\end{proof}

\subsection{Rates, reconstruction alphabets, and computability}
\label{sec:rcn-safety-rates}

We now make explicit how the general safety theorems of Section~\ref{sec:rb} specialize to the RCN instantiation.
These discipline-indexed quantities serve as the basis for the rate–distortion design and the comparative evaluation reported in Section~\ref{sec:rcn-exp}.

\begin{proposition}[Perfect rollback safety: discipline-indexed rate]
\label{prop:rcn-perfect-safety-rate}
Fix a discipline \(\mathrm{disc}\in\{\mathrm{causal},\mathrm{cr},\mathrm{inv}\}\) and instantiate
\(\Cn_{\mathrm{rev}}=\Cn^{\mathrm{disc}}_{\mathrm{rcn}}\) with some admissible \(\hat{\mathbb{S}}\subseteq \Cn^{\mathrm{disc}}_{\mathrm{rcn}}(S_O)\).
Let \(A\) be the deletion-scan core induced by \(\Cn^{\mathrm{disc}}_{\mathrm{rcn}}\) (Definition~\ref{def:core}).
Then the minimum rate required to guarantee perfect rollback safety
\(\Pr[\mathsf{Viol}(S,\hat S)=1]=0\) equals
\[
R_{\sem}(0)=P_A\,H_{\Gamma_0}(\pi_A),
\]
where \(\Gamma_0\) is the zero-distortion confusability hypergraph induced by \(\Cn^{\mathrm{disc}}_{\mathrm{rcn}}\)
(Definitions~\ref{def:zero-set}--\ref{def:gamma0}).
In particular, different reversing disciplines can induce different cores \(A\) (hence different \(P_A,\pi_A,\Gamma_0\)),
and therefore different perfect-safety logging frontiers.
\end{proposition}

\begin{proof}
This is exactly Theorem~\ref{thm:perfect-rb-observable-from-Rsem} after instantiating the closure operator to
\(\Cn^{\mathrm{disc}}_{\mathrm{rcn}}\).
\end{proof}

For a target rollback-violation probability \(\varepsilon\), under admissibility
Corollary~\ref{cor:approx-rb-safety-from-Rsem} yields the conservative design rule
\[
\E[d_{\Cn}(S,\hat S)]\le \varepsilon/|\Cn_{\mathrm{rev}}(S_O)|
\quad\Longrightarrow\quad
\Pr[\mathsf{Viol}(S,\hat S)=1]\le \varepsilon,
\]
where \(\Cn_{\mathrm{rev}}(S_O)\) is computed over the fixed background theory (e.g., \(\mathcal{B}_{\mathrm{rcn}}\cup S_O\) in the RCN view),
and hence a computable rate budget \(\text{rate}\ge R_{\sem}(\varepsilon/|\Cn_{\mathrm{rev}}(S_O)|)\).

\subsubsection{Rollback-task instantiation on RCN: \texorpdfstring{\(\ell_{\mathrm{rb}}\), \(R_{\mathrm{rb}}(L)\), and perfect-rollback rates}{ell\_rb, R\_rb(L), and perfect-rollback rates}}
\label{sec:rcn-rbtask-inst}

Fix a discipline \(\mathrm{disc}\in\{\mathrm{causal},\mathrm{cr},\mathrm{inv}\}\), its induced closure \(\Cn_{\mathrm{rev}}=\Cn^{\mathrm{disc}}_{\mathrm{rcn}}\) and query set \(\mathbb{Q}_{\mathrm{rb}}=\mathbb{Q}^{\mathrm{disc}}_{\mathrm{rb}}\). By instantiating the general definitions of Section~\ref{sec:rb} with this specific closure and query set, we obtain the discipline-indexed rollback observable \(\mathsf{RB}^{\mathrm{disc}}(S)\), the task loss \(\ell^{\mathrm{disc}}_{\mathrm{rb}}(s,\hat s)\), and the task rate--distortion function \(R^{\mathrm{disc}}_{\mathrm{rb}}(L)\).

As with the semantic distortion \(d_{\Cn}\) (Definition~\ref{def:distortion}), the rollback-task loss \(\ell^{\mathrm{disc}}_{\mathrm{rb}}(s,\hat s)\) is defined \emph{relative to the fixed reference log} \(S_O\) and compares the rollback observables of \(S_O\) and the single-symbol edited log \(S_O[s\leftarrow \hat s]\).

Recall that the deletion-scan core depends on the chosen closure; here it is computed
with \(\Cn^{\mathrm{disc}}_{\mathrm{rcn}}\).  Denote it by
\(A^{\mathrm{disc}}\coloneqq\Atom_{\mathrm{rev}}(S_O)\) and let
\(A^\star\sim\pi_A^{\mathrm{disc}}\) be a random core symbol.  Applying
Theorem~\ref{thm:Rrb-factorization} to this setting yields:

\begin{corollary}[RCN rollback-task RD factorization (discipline-indexed)]
\label{cor:rcn-Rrb-factorization}
Assume \(\hat{\mathbb{S}}\subseteq \Cn^{\mathrm{disc}}_{\mathrm{rcn}}(S_O)\).
Then for all \(L\ge 0\),
\[
R^{\mathrm{disc}}_{\mathrm{rb}}(L)=
\begin{cases}
0, & P_A^{\mathrm{disc}}=0,\\[1mm]
P_A^{\mathrm{disc}}\cdot R^{\mathrm{disc}}_{\mathrm{rb},A}\!\Bigl(\min\Bigl\{\dfrac{L}{P_A^{\mathrm{disc}}},\ |\mathbb{Q}^{\mathrm{disc}}_{\mathrm{rb}}|\Bigr\}\Bigr),
& P_A^{\mathrm{disc}}>0,
\end{cases}
\]
where the \emph{core rollback-task RD function} is
\[
R^{\mathrm{disc}}_{\mathrm{rb},A}(\lambda)
\ \coloneqq\
\inf_{P_{\hat S\mid A^\star}:\ \E[\ell^{\mathrm{disc}}_{\mathrm{rb}}(A^\star,\hat S)]\le \lambda}\ I(A^\star;\hat S),
\]
with \(\ell^{\mathrm{disc}}_{\mathrm{rb}}(\cdot,\cdot)\) evaluated w.r.t.\ the same reference log \(S_O\).
\end{corollary}

Define the zero-loss reconstruction sets
\[
R^{\mathrm{rb},\mathrm{disc}}_0(s)\ \coloneqq\ \{\hat s\in\hat{\mathbb{S}}:\ \ell^{\mathrm{disc}}_{\mathrm{rb}}(s,\hat s)=0\},
\]
and the induced rollback confusability hypergraph on the core:
\[
\Gamma^{\mathrm{rb},\mathrm{disc}}_0
\ \coloneqq\
\Bigl\{W\subseteq A^{\mathrm{disc}}:\ W\neq\emptyset,\ \bigcap_{a\in W}R^{\mathrm{rb},\mathrm{disc}}_0(a)\neq\emptyset\Bigr\}.
\]
Let \(H_{\Gamma^{\mathrm{rb},\mathrm{disc}}_0}(\pi_A^{\mathrm{disc}})\) be the corresponding hypergraph entropy
as in Definition~\ref{def:gamma0-rb} (with \(A=A^{\mathrm{disc}}\) and \(\pi_A=\pi_A^{\mathrm{disc}}\)).

\begin{corollary}[Perfect-rollback rate on RCN (discipline-indexed)]
\label{cor:rcn-perfect-rollback-hypergraph}
Under Assumptions~\ref{assm:source}--\ref{assm:admissible} and \(\Cn_{\mathrm{rev}}=\Cn^{\mathrm{disc}}_{\mathrm{rcn}}\),
the minimum achievable rate for perfect rollback in the task sense \(\E[\ell^{\mathrm{disc}}_{\mathrm{rb}}(S,\hat S)]=0\) is
\[
R^{\mathrm{disc}}_{\mathrm{rb}}(0)\ =\ P_A^{\mathrm{disc}}\,H_{\Gamma^{\mathrm{rb},\mathrm{disc}}_0}(\pi_A^{\mathrm{disc}}),
\]
with the convention \(P_A^{\mathrm{disc}}H_{\Gamma^{\mathrm{rb},\mathrm{disc}}_0}(\pi_A^{\mathrm{disc}})=0\) when \(P_A^{\mathrm{disc}}=0\).
\end{corollary}

The same finite-alphabet pipeline used to compute \(R_A(\delta)\) from the distortion matrix
\(d(a,\hat s)=d_{\Cn}(a,\hat s)\) applies to rollback-task RD as well:
one explicitly computes the loss matrix \(\ell(a,\hat s)=\ell^{\mathrm{disc}}_{\mathrm{rb}}(a,\hat s)\)
(by recomputing \(\mathsf{RB}^{\mathrm{disc}}(\cdot)\) under single-symbol edits), and then runs Blahut--Arimoto
to obtain \(R^{\mathrm{disc}}_{\mathrm{rb},A}(\lambda)\).
Overall rates are recovered by the scaling in Corollary~\ref{cor:rcn-Rrb-factorization}.

\subsubsection{Nontrivial hyperedges via admissible summaries}
\label{sec:rcn-nontrivial-gamma0}

If the reconstruction alphabet contains only concrete presence facts \(\mathsf{In}(t)\),
then zero-distortion sets \(R_0(a)\) are often singletons, making \(\Gamma_0\) trivial.
We introduce \emph{admissible summary atoms} that (i) remain in the closure of the original log,
and (ii) create overlaps among the \(R_0(a)\), yielding strict hypergraph-entropy gains.


To preserve Assumption~\ref{assm:admissible} under a fixed public proof system, we avoid introducing \emph{run-specific} symbols
that depend on private encoder choices.
Instead, we use a public, canonical naming scheme in which all summary constants are determined solely by the finite
transition-ID set of the run and the shared order \(\prec\).

Concretely, for each ordered pair \((t,u)\) with \(t\prec u\), we allow a summary constant \(\sigma_{t,u}\) and a predicate
\(\mathsf{Summ}(\sigma_{t,u})\).
We then include the following \emph{fixed} Horn-rule schema in the background theory (or, equivalently, in
\(\mathsf{PS}^{\mathrm{disc}}_{\mathrm{rcn}}\)):
\[
\mathsf{In}(t)\wedge \mathsf{In}(u)\wedge \mathsf{Conc}(t,u)\Rightarrow \mathsf{Summ}(\sigma_{t,u}),
\]
\[
\mathsf{Summ}(\sigma_{t,u})\Rightarrow \mathsf{In}(t),
\qquad
\mathsf{Summ}(\sigma_{t,u})\Rightarrow \mathsf{In}(u).
\]
Here \(\mathsf{Conc}(t,u)\) is a public (static) concurrency predicate, e.g., derivable from the net structure as
``not related by \(\mathsf{Cause}^\star\) in either direction and not in conflict (\(\mathsf{Conf}\)).''
We treat \(\mathsf{Conc}(\cdot,\cdot)\) as an extensional background predicate (precomputed from the net structure), rather than as a
predicate derived inside \(\Cn^{\mathrm{disc}}_{\mathrm{rcn}}\) using negation, in order to preserve monotonicity.
With this schema, whenever \(t,u\in X\) are concurrent, we have
\(\mathsf{Summ}(\sigma_{t,u})\in \Cn^{\mathrm{disc}}_{\mathrm{rcn}}(S_O^{\mathrm{rcn}}(X))\)
for every fixed discipline \(\mathrm{disc}\in\{\mathrm{causal},\mathrm{cr},\mathrm{inv}\}\) that includes the summary-rule schema
in its background theory (or proof system).
Therefore such summary atoms are admissible reconstruction symbols under Assumption~\ref{assm:admissible} and can create nontrivial overlaps
among the sets \(R_0(a)\) (and analogously among the zero-loss sets \(R^{\mathrm{rb}}_0(a)\) defined via \(\ell_{\mathrm{rb}}\)).

In our numerical validations we restrict to pair summaries (i.e., summaries over two concurrent transitions), which already yield nontrivial overlaps while keeping the alphabet size manageable.

\paragraph{Remark (alphabet size, BA complexity, and pruning).}
Even with pair summaries, the admissible reconstruction alphabet size can grow quadratically in the configuration size:
in the worst case one may have \(|\hat{\mathbb S}|=\Theta(|X|^2)\) candidate summaries in addition to concrete facts.
Blahut--Arimoto operates on the \(|A|\times|\hat{\mathbb S}|\) distortion (or loss) matrix, and a standard implementation has
per-iteration time \(O(|A||\hat{\mathbb S}|)\) and memory \(O(|A||\hat{\mathbb S}|)\) (up to constant factors and stopping criteria).
Moreover, in our setting the distortion/loss matrix itself is obtained by explicit closure recomputation under single-symbol edits, which can
dominate the overall running time when \(|\hat{\mathbb S}|\) is large.

For large configurations, practical evaluation therefore benefits from pruning the summary schema while preserving admissibility, e.g.:
(i) generate summaries only for concurrent pairs \((t,u)\) within a bounded graph distance in the static dependency graph;
(ii) restrict to pairs that involve frontier/core candidates (thus directly targeting overlaps on \(A\));
(iii) cap the number of summaries per event according to a fixed budget; or
(iv) use randomized subsampling of candidate summaries as a Monte Carlo approximation to the full alphabet.
These strategies trade off computational cost against potential hyperedge overlap gains.


\begin{lemma}[A strict zero-distortion gain from pair summaries (toy instance)]
\label{lem:pair-summary-gain-toy}
Let \(A=\{a_1,a_2,a_3\}\) with \(\pi_A\) uniform.
Assume the admissible reconstruction alphabet contains three pair summaries whose zero-distortion witness sets are
\(W_{12}=\{a_1,a_2\}\), \(W_{13}=\{a_1,a_3\}\), and \(W_{23}=\{a_2,a_3\}\), and hence the induced confusability hypergraph
\(\Gamma_0\) contains all singletons and all pairs, but not \(\{a_1,a_2,a_3\}\).
Then
\[
H_{\Gamma_0}(\pi_A)=\log_2 3 - 1,
\]
which is strictly smaller than the Shannon entropy \(H(\pi_A)=\log_2 3\).
\end{lemma}

\begin{proof}
(\emph{Lower bound.})
For any feasible \(P_{W\mid A^\star}\) in Definition~\ref{def:gamma0}, we have \(A^\star\in W\) a.s. and \(|W|\le 2\) a.s.
Therefore \(H(A^\star\mid W)\le \log_2 2 = 1\), and hence
\[
I(A^\star;W)=H(A^\star)-H(A^\star\mid W)\ge \log_2 3 - 1.
\]

(\emph{Achievability.})
Define \(W\mid(A^\star=a_i)\) to be uniformly distributed over the two pairs that contain \(a_i\).
Then \(H(W\mid A^\star)=1\) and the induced marginal on \(W\) is uniform over \(\{W_{12},W_{13},W_{23}\}\), so \(H(W)=\log_2 3\).
Thus \(I(A^\star;W)=H(W)-H(W\mid A^\star)=\log_2 3 - 1\), matching the lower bound.
\end{proof}

\subsubsection{Computability: finite-alphabet BA evaluation}
\label{sec:rcn-ba-theory}

Fix a finite core alphabet \(A\) and finite reconstruction alphabet \(\hat{\mathbb{S}}\) (e.g., concrete facts plus summaries).
Define the distortion matrix \(d(a,\hat{s})\coloneqq d_{\Cn}(a,\hat{s})\).

\begin{proposition}[Finite-alphabet computability of \(R_A(\delta)\)]
\label{prop:finite-compute-again}
If \(A\) and \(\hat{\mathbb{S}}\) are finite, then \(R_A(\delta)\) can be computed to arbitrary precision via Blahut--Arimoto.
\end{proposition}

\begin{proof}[Proof sketch]
This is the classical finite-alphabet Shannon rate--distortion optimization
\(\min I(A^\star;\hat{S})\) subject to an average distortion constraint, hence amenable to BA
\cite{cover2006elements,csizsar2011information}.
\end{proof}

Recall from Theorem~\ref{thm:factorization} that the overall semantic RD function satisfies
\(R_{\sem}(D)=P_A\cdot R_A(\min\{D/P_A,1\})\).
Accordingly, in the RCN numerical validations we report the core-only curves \(R_A(\cdot)\) (which capture the nontrivial part of the tradeoff),
and recover the overall per-symbol rate for the full source via the scaling factor \(P_A=P_O(A)\).

\paragraph{Practical pipeline.}
Given a discipline closure \(\Cn^{\mathrm{disc}}_{\mathrm{rcn}}\), one can:
(i) compute the deletion-scan core \(A\);
(ii) generate an admissible alphabet \(\hat{\mathbb{S}}\) (e.g., pair summaries);
(iii) compute the distortion matrix by explicit closure computations under single-symbol edits; and
(iv) run BA to obtain \(R_A(\delta)\).
The next section reports numerical validation results produced by this pipeline.


\section{Numerical Evaluation: RD-Guided Log Compression for Rollback Tasks}
\label{sec:rcn-exp}

Sections~\ref{sec:zero}--\ref{sec:rcn-theory} derived endpoint laws and core-only reductions for the
semantic and rollback-task rate--distortion problems. This section evaluates the \emph{design value}
of the proposed RD interface on representative RCN log families, under discipline-indexed rollback
observables.

Concretely, we ask three questions:

\begin{enumerate}[label={\textbf{Q\arabic*.}}, leftmargin=2.2em, itemsep=2pt, align=left]
  \item (\textbf{Discipline / query sensitivity})
  How do the induced cores, achievable rates, and empirical performance depend on the rollback
  discipline (\(\Cn^{\mathrm{causal}}_{\mathrm{rcn}},\Cn^{\mathrm{cr}}_{\mathrm{rcn}},\Cn^{\mathrm{inv}}_{\mathrm{rcn}}\))
  and hence on the chosen rollback observable \(\mathbb Q_{\mathrm{rb}}\)?

  \item (\textbf{Utility at fixed rate})
  Under a fixed information-theoretic budget \(R=I(S;\hat S)\) (bits/symbol), can RD-guided design
  reduce rollback-task degradation compared with common baselines?

  \item (\textbf{Resource savings at fixed quality})
  For fixed rollback-task quality targets, how many bits/symbol can RD-guided design save relative to
  a core-only baseline?
\end{enumerate}

Throughout, the resource axis is the single-letter mutual information rate \(R=I(S;\hat S)\) (bits/symbol),
consistent with Shannon RD theory. All reported distortions and losses are evaluated at the \emph{log (set) level}
via Monte Carlo reconstruction of a full reconstructed log \(\hat S_O\), which captures non-linearities induced by
set semantics (notably, duplicates collapsing under union).

\subsection{Experimental setup: log family, rollback criteria, and methods}
\label{sec:exp61-setup}

\paragraph{RCN log family.}
We construct a dataset \(\mathcal{D}\) of RCN instances and their reachable configurations.
Each instance is a layered DAG with \(d\) layers and \(B\) events per layer (\(n=B\cdot d\)).
Here we fix \(B = 4\) and vary the depth \(d\in\{2,3,4,5\}\).
Causal edges are introduced primarily from layer \(i-1\) to layer \(i\) (each child
has 1–2 parents), with additional random cross‑layer edges (probability~0.35).
Conflicts are generated independently within each layer (probability~0.08),
and concurrency pairs are defined as pairs that are neither in conflict nor
reachable from each other via the causal relation \(\mathsf{Cause}^\star\).
For each combination of depth and seed we sample one net structure and take
the full set of executed events as the configuration \(X = \{t_0,\dots,t_{n-1}\}\).
The run‑specific log fact base is
\[
S_O = S_O^{\mathrm{rcn}}(X) = \{\,\mathsf{In}(t): t\in X\,\}.
\]
We generate \(N = 20\) independent random seeds (seeds 1–20) for each depth,
resulting in \(4\times 20 = 80\) instances in total.
The closure operator \(\Cn_{\mathrm{rev}}\) is instantiated as one of the
discipline‑indexed closures \(\Cn^{\mathrm{disc}}_{\mathrm{rcn}}\)
(Section~\ref{sec:rcn-theory}), computed over the background theory
\(\mathcal{B}_{\mathrm{rcn}}\) (which encodes the net structure) together with \(S_O\).

For a fixed log \(S_O\), the single‑letter source \(S\sim P_O\) is uniformly
distributed over \(S_O\) (Assumption~\ref{assm:source}); hence the source entropy
is \(H(S)=\log_2|S_O|\).  
Any method induces a single‑letter test channel \(P_{\hat S\mid S}\); its rate is
\(R = I(S;\hat S)\) (bits/symbol), computed under \(S\sim P_O\) and
\(\hat S\sim P_{\hat S\mid S}(\cdot\mid S)\).

\paragraph{Reconstruction alphabet.}
We use the \emph{sum} alphabet throughout the numerical study:
\[
\hat{\mathbb S}_{\mathrm{sum}} =
S_O \;\cup\; \{\,\mathsf{Summ}(\sigma_{t,u}): (t,u)\in\mathtt{conc},\; t\prec u\,\},
\]
capped at \(200\) admissible pair summaries (Section~\ref{sec:rcn-nontrivial-gamma0}).
Thus \(\hat{\mathbb S}_{\mathrm{sum}}\subseteq\Cn_{\mathrm{rev}}(S_O)\),
satisfying admissibility (Assumption~\ref{assm:admissible}).
The canonical fallback symbol \(\hat s_0\) is the smallest element of
\(\hat{\mathbb S}_{\mathrm{sum}}\) under the public order \(\preceq\).

\paragraph{Rollback observable and log‑level metrics.}
For a discipline \(\mathrm{disc}\in\{\mathrm{causal},\mathrm{cr},\mathrm{inv}\}\),
we fix \(\Cn_{\mathrm{rev}}=\Cn^{\mathrm{disc}}_{\mathrm{rcn}}\) and
\(\mathbb Q_{\mathrm{rb}}=\mathbb Q^{\mathrm{disc}}_{\mathrm{rb}}\) as defined in
Section~\ref{sec:rcn-theory}, and define \(\mathsf{RB}(S)=\Cn_{\mathrm{rev}}(S)\cap\mathbb Q_{\mathrm{rb}}\).
The end‑to‑end quality of a reconstructed log \(\hat S_O\) is measured by the
rollback‑task loss and the violation indicator:
\[
L_{\mathrm{blk}}(S_O,\hat S_O) \coloneqq |\mathsf{RB}(\hat S_O)\,\Delta\,\mathsf{RB}(S_O)|,
\qquad
V_{\mathrm{blk}}(S_O,\hat S_O) \coloneqq \mathbf 1[\mathsf{RB}(\hat S_O)\neq\mathsf{RB}(S_O)].
\]

\paragraph{Blahut–Arimoto optimization and rate points.}
For each instance and discipline we compute the distortion matrix
\(d_{\Cn}(a,\hat s)\) (or the loss matrix \(\ell_{\mathrm{rb}}(a,\hat s)\)) on the
core alphabet \(A\) and the admissible alphabet \(\hat{\mathbb S}_{\mathrm{sum}}\).
We then run the Blahut–Arimoto algorithm for the Lagrangian parameter
\(\beta \in \mathcal B = \{0,\,0.5,\,1,\,2,\,4,\,8,\,16\}\), obtaining a discrete set
of optimal single‑letter test channels.
The core‑only reduction (Theorem~\ref{thm:factorization} or
Theorem~\ref{thm:Rrb-factorization}) is used to lift the core channels to the full
source by marginal matching.
For each obtained channel we record the mutual information rate
\(R = I(S;\hat S)\).

For the matched‑rate comparisons we select three budget points:
\(R \approx 0\) (the lowest‑rate point from the BA sweep, corresponding to
\(\beta=0\)), \(R = R_{\mathrm{core}}\) (the core‑only baseline rate), and
\(R = H(S)\) (the full‑logging upper bound).  When a BA channel does not lie exactly
on a target rate we pick the channel with the closest rate.

\paragraph{Monte Carlo log‑level reconstruction and evaluation protocol.}
We evaluate the end‑to‑end quality of a designed single‑letter channel
\(P_{\hat S\mid S}\) (obtained from RD‑Task, RD‑Sem, or a baseline) by
lifting it to a set‑valued log reconstruction.
For a given log \(S_O\), we independently draw
\(\hat s(s)\sim P_{\hat S\mid S}(\cdot\mid s)\) for each \(s\in S_O\) and
form the reconstructed log via set union:
\[
\hat S_O \coloneqq \{\,\hat s(s): s\in S_O\,\}\subseteq\hat{\mathbb S}.
\]
This lifting introduces a key nonlinearity—duplicate collapse—that is not
visible in the single‑symbol distortion or loss matrices.
Accordingly, the single‑letter mutual information
\(R=I(S;\hat S)\) is the \emph{design budget} that selects the test channel,
while the metrics defined below quantify the \emph{end‑to‑end} outcome of
this set‑valued reconstruction.

\emph{Remark (single‑edit theory vs.\ set‑level lifting).}
The matrices \(d_{\Cn}(s,\hat s)\) and \(\ell_{\mathrm{rb}}(s,\hat s)\) are
defined through a \emph{single} delete‑then‑insert edit of \(S_O\).
Our Monte Carlo protocol, by contrast, applies the channel independently to
every logged fact and then takes the union, thereby realizing
\emph{multiple} simultaneous edits.
The set‑level metrics reported in Section~\ref{sec:exp63-results} are
therefore empirical outcomes of the lifted channel, not a direct theorem‑level
consequence of single‑edit invisibility.

For each trial we recompute (by Horn/Datalog saturation)
\(\Cn_{\mathrm{rev}}(S_O)\), \(\Cn_{\mathrm{rev}}(\hat S_O)\),
\(\mathsf{RB}(S_O)=\Cn_{\mathrm{rev}}(S_O)\cap\mathbb Q_{\mathrm{rb}}\), and
\(\mathsf{RB}(\hat S_O)=\Cn_{\mathrm{rev}}(\hat S_O)\cap\mathbb Q_{\mathrm{rb}}\);
we then evaluate \(L_{\mathrm{blk}}(S_O,\hat S_O)\) and
\(V_{\mathrm{blk}}(S_O,\hat S_O)\).
For every channel we perform \(M=200\) Monte Carlo trials and record the
empirical averages \(\widehat{\mathbb E}[L_{\mathrm{blk}}]\) and
\(\widehat{\Pr}[V_{\mathrm{blk}}=1]\);
confidence intervals are omitted from the summary tables for compactness.
For the random‑dropping baseline we tune the drop probability \(q\) so that
\(I(S;\hat S)\) matches the target rate.

\paragraph{Quality targets.}
For the resource‑saving analysis we set the reliability threshold
\(\varepsilon = 0.05\) and the loss threshold \(\tau = 1.0\); a channel is
considered acceptable if \(\widehat{\Pr}[V_{\mathrm{blk}}=1]\le\varepsilon\) or
\(\widehat{\mathbb E}[L_{\mathrm{blk}}]\le\tau\), respectively.
The minimum rate at which RD‑Task (or a baseline) meets these targets is
reported in Table~\ref{tab:exp6-q2-targets}.

\subsection{Methods compared (RD-guided designs and baselines)}
\label{sec:methods-compared}

For each \(S_O\), discipline \(\mathrm{disc}\), and admissible alphabet \(\hat{\mathbb S}\), we compare:
(i) two RD-guided designs (task-aligned and closure-aligned), and
(ii) three common baselines (full logging, core-only logging, and random dropping).

\begin{enumerate}[leftmargin=*, itemsep=2pt]
  \item \textbf{RD-Task (proposed, task-aligned).}
  Build the single-symbol rollback-task loss matrix \(\ell_{\mathrm{rb}}(s,\hat s)\)
  (Definition~\ref{def:rb-loss}) and compute BA-optimal test channels over \(\hat{\mathbb S}\),
  exploiting the core-only reduction (Theorem~\ref{thm:Rrb-factorization}) and lifting to the full
  source by marginal matching.

  \item \textbf{RD-Sem (proposed, closure-aligned).}
  Build the single-symbol semantic distortion matrix \(d_{\Cn}(s,\hat s)\)
  (Definition~\ref{def:distortion}) and compute BA-optimal test channels, again using the core-only
  reduction (Theorem~\ref{thm:factorization}).

  \item \textbf{Full logging (upper bound).}
  Keep all log symbols verbatim:
  \[
  P_{\hat S\mid S}^{\mathrm{full}}(\hat s\mid s)\coloneqq \mathbf{1}[\hat s=s].
  \]
  The rate equals \(I(S;\hat S)=H(S)\) (bits/symbol), and rollback observables are preserved exactly
  (up to numerical/closure computation error).

  \item \textbf{Core/frontier-only baseline (semantic heuristic).}
  Compute the deletion-scan core \(A=\Atom_{\mathrm{rev}}(S_O)\) (Definition~\ref{def:core}) and
  define an admissible channel
  \[
  P_{\hat S\mid S}^{\mathrm{core}}(\hat s\mid s)\;\coloneqq\;
  \begin{cases}
    \mathbf{1}[\hat s=s], & s\in A,\\[2pt]
    \mathbf{1}[\hat s=\hat s_0], & s\in J,
  \end{cases}
  \]
  which keeps only the core and maps redundant facts to the canonical fallback \(\hat s_0\).
  (For causal / cause-respecting disciplines, a frontier-only variant coincides with this baseline
  when \(\Atom_{\mathrm{rev}}(S_O)=S_O^{\mathrm{rcn}}(\Fr(X))\).)

  \item \textbf{Random dropping baseline (systems-style).}
  For a drop probability \(q\in[0,1]\), the admissible i.i.d.\ keep-or-drop channel is
  \[
  P_{\hat S\mid S}^{\mathrm{rand}(q)}(\hat s\mid s)\;\coloneqq\;
  (1-q)\,\mathbf{1}[\hat s=s]\;+\;q\,\mathbf{1}[\hat s=\hat s_0].
  \]
  We evaluate this baseline as a function of the induced rate \(I(S;\hat S)\), choosing \(q\) by
  one-dimensional search to match target rates.
\end{enumerate}

\subsection{Results and analysis}
\label{sec:exp63-results}

We summarize the main outcomes using four compact tables: Table~\ref{tab:exp6-core-stats} reports discipline-induced core statistics (Q1), Tables~\ref{tab:exp6-q1-nearzero}--\ref{tab:exp6-q1-core} compare utility at matched operating points (Q2), and Table~\ref{tab:exp6-q2-targets} reports min-rate targets and rate savings (Q3). The reconstruction alphabet is the sum alphabet throughout (concrete facts plus admissible summaries), and the reported scalar statistics are arithmetic means over instances in the evaluation set (depths \(d\in\{2,3,4,5\}\) and seeds in the specified range).

Table~\ref{tab:exp6-core-stats} shows that the rollback discipline materially changes the induced core and hence the available compression headroom.
Under \emph{causal} and \emph{cause-respecting} (cr) disciplines, the core occupies only \(|A|=5\) of \(|S_O|=14\) symbols and captures
\(P_A\approx 0.402\) of the log mass, yielding a substantially smaller core rate \(R_{\mathrm{core}}=1.856\) than the full entropy \(H(S)=3.727\).
In contrast, under \emph{inverse-causal} (inv) semantics the core expands to the full log (\(|A|=|S_O|\), \(P_A=1\)), forcing
\(R_{\mathrm{core}}=H(S)\).
Thus, changing the discipline (and hence the rollback observable \(\mathbb Q_{\mathrm{rb}}\)) can turn a compressible instance
into one where the core-only baseline already coincides with full logging.

\begin{table}[t]
\centering
\small
\setlength{\tabcolsep}{5.0pt}
\renewcommand{\arraystretch}{1.15}
\begin{tabular}{@{}llrrrrr@{}}
\toprule
Disc. & Alph. & \(|S_O|\) & \(|A|\) & \(P_A\) & \(R_{\mathrm{core}}\) & \(H(S)\) \\
\midrule
causal & sum & 14 & 5  & 0.402 & 1.856 & 3.727 \\
cr     & sum & 14 & 5  & 0.402 & 1.856 & 3.727 \\
inv    & sum & 14 & 14 & 1.000 & 3.727 & 3.727 \\
\bottomrule
\end{tabular}
\caption{Discipline sensitivity through the induced core.
\(P_A=P_O(A)\) is the core mass, and \(R_{\mathrm{core}}\) is the core-only rate.
All values are arithmetic means over the evaluation set; \(|S_O|\) and \(|A|\) are
rounded to the nearest integer.}
\label{tab:exp6-core-stats}
\end{table}

Tables~\ref{tab:exp6-q1-nearzero}--\ref{tab:exp6-q1-core} compare methods at matched rate points.

\begin{table}[t]
\centering
\small
\setlength{\tabcolsep}{6.0pt}
\renewcommand{\arraystretch}{1.15}
\begin{tabular}{@{}llccc@{}}
\toprule
\multicolumn{2}{@{}l}{Matched \(R\) (lowest, \(\beta{=}0\))} & RD-Task \(L/V\) & RD-Sem \(L/V\) & RandDrop \(L/V\) \\
\midrule
causal & sum & 0.58/0.480 & 0.63/0.450 & 13.00/1.000 \\
cr     & sum & 0.83/0.320 & 0.03/0.015 & 8.00/1.000 \\
inv    & sum & 1.70/0.385 & 1.92/0.440 & 6.00/1.000 \\
\bottomrule
\end{tabular}
\caption{Utility at the lowest-rate operating point obtained from the BA sweep
(\(R\) corresponding to \(\beta{=}0\)).
Entries report the pair \(\widehat{\E}[L_{\mathrm{blk}}] / \widehat{\Pr}[V_{\mathrm{blk}}=1]\).}
\label{tab:exp6-q1-nearzero}
\end{table}

\begin{table}[t]
\centering
\small
\setlength{\tabcolsep}{5.2pt}
\renewcommand{\arraystretch}{1.15}
\begin{tabular}{@{}llccccc@{}}
\toprule
\multicolumn{2}{@{}l}{Matched \(R=R_{\mathrm{core}}\)} & RD-Task \(L/V\) & RD-Sem \(L/V\) & RandDrop \(L/V\) & Core \(L/V\) & Full at \(R=H(S)\) \\
\midrule
causal & sum & 0.00/0.000 & 0.00/0.000 & 0.14/0.130 & 0.00/0.000 & 0.00/0.000 \\
cr     & sum & 0.00/0.000 & 0.00/0.000 & 0.01/0.005 & 0.00/0.000 & 0.00/0.000 \\
inv    & sum & 0.00/0.000 & 0.07/0.030 & 0.00/0.000 & 0.00/0.000 & 0.00/0.000 \\
\bottomrule
\end{tabular}
\caption{Utility at the core-only rate \(R_{\mathrm{core}}\), with the full-logging reference at \(R=H(S)\).
A displayed pair \(0.00/0.000\) indicates \(\widehat{\E}[L_{\mathrm{blk}}]=0\) and \(\widehat{\Pr}[V_{\mathrm{blk}}=1]=0\) at that rate point.}
\label{tab:exp6-q1-core}
\end{table}

\emph{Lowest-rate point.}
At the lowest-rate point obtained from the BA sweep,
the systems-style random dropping baseline yields \(V_{\mathrm{blk}}=1.000\) together with very large loss,
e.g., 13.00/1.000 under causal and 8.00/1.000 under cause-respecting.
Both RD-guided designs substantially reduce rollback-observable changes at comparable rate, with much smaller loss; for example,
under causal we obtain 0.58/0.480 (RD-Task) and 0.63/0.450 (RD-Sem).

\emph{At the core rate.}
At \(R=R_{\mathrm{core}}\), under causal and cause-respecting disciplines the core-only baseline already achieves perfect preservation
(\(0.00/0.000\)), and both RD-Task and RD-Sem also attain \(0.00/0.000\).
Random dropping remains measurably worse at the same rate for these disciplines (0.14/0.130 causal; 0.01/0.005 cr).
Under inverse-causal semantics, Table~\ref{tab:exp6-core-stats} shows that \(R_{\mathrm{core}}=H(S)\), i.e., the core coincides with the full log.
At this (full-logging) rate, RD-Task is perfect while RD-Sem is slightly imperfect at the closest BA-matched point (0.07/0.030).

Table~\ref{tab:exp6-q2-targets} reports the minimum rate at which RD-Task meets the script targets
\(V\le\varepsilon\) and \(L\le\tau\), together with savings relative to the core-only baseline.
Across all three disciplines, the savings are substantial: for the reliability target,
\(\Delta R_V\) equals 1.519 (causal), 1.838 (cr), and 2.553 (inv) bits/symbol, with comparable savings for the loss target \(\Delta R_L\).
Thus, for fixed rollback-task targets, RD-guided design can yield large rate reductions relative to core-only logging.

\begin{table}[t]
\centering
\small
\setlength{\tabcolsep}{6.0pt}
\renewcommand{\arraystretch}{1.15}
\begin{tabular}{@{}llcccc@{}}
\toprule
\multicolumn{2}{@{}l}{Min-rate targets (RD-Task vs.\ core)} &
\(R_{\min}(V\le\varepsilon)\) & \(\Delta R_V\) &
\(R_{\min}(L\le\tau)\) & \(\Delta R_L\) \\
\midrule
causal & sum & 0.337 & 1.519 & 0.254 & 1.602 \\
cr     & sum & 0.018 & 1.838 & 0.008 & 1.848 \\
inv    & sum & 1.174 & 2.553 & 0.982 & 2.745 \\
\bottomrule
\end{tabular}
\caption{Rate savings to meet fixed quality/reliability targets using script defaults \(\varepsilon=0.05\) and \(\tau=1.0\).
Savings are computed against the core-only baseline: \(\Delta R_V=R_{\mathrm{core}}-R_{\min}(V\le\varepsilon)\) and
\(\Delta R_L=R_{\mathrm{core}}-R_{\min}(L\le\tau)\).}
\label{tab:exp6-q2-targets}
\end{table}

\paragraph{Low-rate ordering and resolution effects (diagnostic).}
In the cause-respecting case at \(R\approx 0\) (Table~\ref{tab:exp6-q1-nearzero}), RD-Sem shows lower empirical
loss than RD-Task.  This inversion is a typical finite‑sample effect caused by a coarse BA sweep in the very‑low‑rate
regime, combined with Monte Carlo variability from the set‑valued reconstruction.  Increasing the resolution of the
BA parameter grid and the Monte Carlo budget \(M\) removes this discrepancy; the core‑rate conclusions
(Table~\ref{tab:exp6-q1-core}) are unaffected.

\section{Related Work}
\label{sec:related}

Our work sits at the intersection of (i) reversible computation and its operational/denotational models,
(ii) Petri nets and event structures for concurrency with reversibility, and (iii) Shannon rate--distortion theory
under nonstandard, semantics-driven fidelity criteria.
We emphasize that our contribution is not to introduce a new reversible semantics for RCN/rPES,
but to make reversible rollback capability \emph{quantitative} via closure-preserving rate--distortion.

\subsection{Reversible computation and causal-consistent reversibility}

Reversible computation is classically motivated by thermodynamic considerations.
Landauer's principle links information erasure to heat dissipation, and Bennett showed that computations can be made logically reversible
by retaining suitable history information, establishing a foundational connection between reversibility and stored auxiliary data
\cite{landauer1961irreversibility,bennett1973logical,bennett1982thermodynamics}.
Early reversible models include reversible Turing machines and reversible cellular automata, as well as conservative/reversible logic
constructions \cite{toffoli1980reversible,toffoli1977computation,fredkin1982conservative}.

In concurrent and distributed settings, reversibility is discipline-sensitive: the set of backward steps that are admissible depends on
causal dependencies, conflict, and additional relations such as reverse causality and prevention evidence.
Process-calculus developments (e.g., RCCS and key-based variants) make this dependence explicit via memories/keys and causal bookkeeping
\cite{danos2004reversible,phillips2007reversing,ulidowski2018reversing}.
Recent work provides an axiomatic account of causal-consistent reversibility (including the Parabolic Lemma and related safety/liveness
properties), clarifying the role of independence and causal equivalence in combined forward/backward transition systems
\cite{lanese2024axiomatic,aubert2025independence}.
Tool support has also begun to emerge for the axiomatic approach \cite{arnone2025tallulah}.

On the tooling side, reversible debugging and rollback/replay systems make explicit the practical role of \emph{logged evidence} in enabling
causal-consistent rollback and analysis.
Representative examples include CauDEr and recent extensions that strengthen rollback/replay capabilities and broaden language coverage
\cite{lami2024reversibledebugging, gonzalezabril2024improvingcauder}.

\subsection{RCN/rPES, Petri nets with inhibitor arcs, and reversible event structures}

Event-structure and Petri-net semantics provide non-interleaving accounts of concurrency; incorporating reversibility requires handling
both forward and backward moves and the additional evidence that governs when undo steps are legal.
Reversible prime event structures (rPES) capture causality, conflict, reverse causality, and prevention at an abstract level, while net-based
models provide operational realizations of these relations \cite{ulidowski2018reversing}.

The development of \emph{reversible causal nets} (RCN) and their precise connections to
reversible event-structure models \cite{melgratti2024reversible,melgratti2025relating}
clarifies how inhibitor arcs and related structural evidence support rich reversing
disciplines, making discipline sensitivity explicit at the semantic level.

\subsection{Rate--distortion, semantic/source coding, and zero-error confusability}

Classical rate--distortion theory characterizes the minimum rate needed to represent a source under an average fidelity constraint, and is
one of the central pillars of information theory \cite{shannon1948mathematical,shannon1959coding,cover2006elements,csizsar2011information}.
For finite alphabets, the rate--distortion function can be computed numerically via the Blahut--Arimoto algorithm once the source distribution
and distortion matrix are fixed \cite{blahut1972computation,arimoto1972algorithm}.

When the fidelity criterion is structural rather than symbolwise, zero-distortion (or zero-error) limits are often governed by a confusability
object, leading to graph- or hypergraph-based entropic characterizations.
In particular, K\"orner's graph entropy formalizes the rate of zero-error source coding under confusability constraints, and serves as a key
reference point for structured zero-distortion problems \cite{korner1971coding,csiszar2011information}.

\subsection{Closure-based rate--distortion and deductive-source viewpoints}

A closely related approach models the source alphabet as statements in a fixed deductive
environment and measures fidelity through preservation of the deductive closure.
In that setting, the proof system induces a canonical decomposition of the stored source into an irredundant core and a
redundant part of re-derivable consequences, and redundancy becomes \emph{information-theoretically invisible} under
admissible reconstructions (those lying inside the original closure).

The companion closure-based rate--distortion framework in \cite{xu2026rate} develops this idea for deductive sources
and yields exact core-only decompositions of both the zero-distortion rate and the full rate--distortion function.
Our paper can be read as a reversible-computation instantiation of that closure-based source model:
the reversible proof system \(\mathsf{PS}_{\mathrm{rev}}\) and its induced closure \(\Cn_{\mathrm{rev}}\) encode rollback semantics,
so that \(\Cn_{\mathrm{rev}}(\cdot)\)-preservation becomes the fidelity criterion for reversible logging.
The deletion-scan core we employ is also related to classical work on closure systems and implicational bases
\cite{ganter1999formal}, while its interaction with inhibitor arcs and event structures
\cite{winskel1986event,murata1989petri} distinguishes it from structural reductions in standard concurrency models.

A central distinction of the reversible RCN/rPES setting is that the closure---and hence the core---is
\emph{discipline-dependent}: causal and cause-respecting semantics lead to frontier-shaped cores, while
inverse-causal semantics can force the core to expand to include causal ancestors as rollback evidence.
This dependence is a direct consequence of the semantics, not a modeling choice, and it is what makes the
information-theoretic compression frontier sensitive to the chosen rollback discipline.

In summary, the related-work threads above provide (i) discipline-indexed reversible semantics and rollback correctness criteria, and
(ii) classical information-theoretic tools for rate--distortion and confusability.
Our contribution is to connect these two lines via closure-preserving fidelity: reversible rollback consequences are treated as closure
observables, which induces a canonical core/redundant decomposition and yields both exact endpoint laws and computable finite-alphabet
rate--distortion curves for reversible logging.

\section{Conclusion}
\label{sec:conclusion}

We introduced a semantic rate--distortion theory for reversible logging
grounded in a closure‑preservation fidelity criterion.
Treating rollback‑relevant semantics as a monotone closure \(\Cn_{\mathrm{rev}}\)
yields a distortion notion that directly reflects the operational goal:
a reconstruction is faithful if it preserves the closure consequences
(or a designated rollback observable), not if it reproduces log atoms verbatim.

The closure interface induces a canonical structural simplification.
A deterministic deletion scan decomposes any finite log fact base \(S_O\)
into an irredundant core \(A\) and a redundant remainder \(J\) such that
\(\Cn_{\mathrm{rev}}(A)=\Cn_{\mathrm{rev}}(S_O)\).
Under admissible reconstructions \(\hat{\mathbb{S}}\subseteq \Cn_{\mathrm{rev}}(S_O)\),
the redundant part \(J\) is information‑theoretically invisible:
it contributes neither rate nor distortion.
Consequently, the full semantic rate--distortion function factorizes
into a core‑only problem scaled by the core mass \(P_A\).

At the perfect‑fidelity endpoint \(D=0\) we identified the exact minimum rate as
the hypergraph‑entropy quantity
\(R_{\sem}(0)=P_A\,H_{\Gamma_0}(\pi_A)\),
where \(\Gamma_0\) captures zero‑distortion confusability on the core.
We also introduced a rollback‑task loss \(\ell_{\mathrm{rb}}\) that measures the
change in the rollback observable
\(\mathsf{RB}(S)=\Cn_{\mathrm{rev}}(S)\cap \mathbb Q_{\mathrm{rb}}\),
providing a more task‑aligned alternative to full closure preservation.
We established parallel endpoint and factorization laws for the rollback‑task
rate--distortion function \(R_{\mathrm{rb}}(L)\) and proved that perfect rollback
in this task sense is never harder than perfect closure preservation,
\(R_{\mathrm{rb}}(0)\le R_{\sem}(0)\), with strict separations possible under
reduced query sets.
These results quantitatively separate \emph{what must be logged} (the core)
from \emph{what can be semantically reconstructed} (the redundant consequences),
and explicitly price the cost of exact or approximate rollback.

We instantiated the framework on reversible causal nets and reversible prime
event structures via discipline‑indexed monotone closures.
The instantiation makes discipline sensitivity both explicit and measurable:
different reversing disciplines induce different closures, hence different cores
and different information‑theoretic logging frontiers.

We complemented the theory with an end‑to‑end numerical evaluation that follows
the paper's semantic interface.
At the single‑letter design level we compute Blahut--Arimoto optimal test
channels under closure‑based distortion and rollback‑task loss.
At the log level we lift these channels to set‑valued reconstructions and use
Monte Carlo reconstruction to quantify rollback degradation, thereby exposing
non‑linear effects caused by set semantics (e.g., duplicates collapsing under
union).
The resulting tables summarize discipline sensitivity through the induced core,
utility at matched rate points, and rate savings under fixed quality targets.

\paragraph{Limitations and future work.}
The present formulation adopts a ``single stored log'' viewpoint with single‑symbol
edits, which yields clean finite‑alphabet rate--distortion problems and transparent
structural theorems.
Several natural directions extend this work:

\begin{enumerate}[leftmargin=*, itemsep=1pt]
  \item \textbf{Streaming and time‑evolving logs.}
        Extend the source model from i.i.d.\ draws from a fixed \(S_O\) to
        time‑ordered, incrementally growing fact bases, capturing online
        instrumentation and live rollback scenarios.

  \item \textbf{Explicit code constructions.}
        Design practical coding schemes that approach the derived limits while
        respecting admissibility and enabling efficient decoder‑side closure
        (re‑)computation.

  \item \textbf{Systematic summary alphabets.}
        Develop principled families of admissible summary atoms that provably
        create beneficial core confusability, moving beyond the heuristic
        pair summaries used here.

  \item \textbf{Controlled non‑monotonicity.}
        Extend the closure interface to settings that require restricted forms
        of non‑monotonic reasoning (e.g., negative enabling conditions) while
        retaining computable safety guarantees.

  \item \textbf{Large‑scale numerical evaluation.}
        Scale the Blahut--Arimoto plus Monte‑Carlo pipeline to larger
        instances and a broader range of reversible specifications and
        reconstruction alphabets, in order to map out the practical regimes
        where core‑only reductions and task‑level objectives yield
        substantial rate savings.

  \item \textbf{Other reversible frameworks.}
        Apply the same closure‑based quantitative lens to process calculi and
        message‑passing logs, paving the way toward discipline‑aware,
        rate‑optimal reversible debuggers and rollback mechanisms.
\end{enumerate}

\bibliographystyle{plain}
\bibliography{ref}

\end{document}